\documentclass[aip,jmp,preprint,amsfonts,amsmath,amssymb]{revtex4-1}

\usepackage[all]{xy}
\usepackage{amsthm}
\usepackage{graphicx}
\usepackage{paralist}
\usepackage[colorlinks=false]{hyperref}

\numberwithin{equation}{section}
\renewcommand{\labelenumi}{(\roman{enumi})}

\theoremstyle{plain}
\newtheorem{theorem}{Theorem}[section]
\newtheorem{corollary}[theorem]{Corollary}
\newtheorem{lemma}[theorem]{Lemma}
\newtheorem{proposition}[theorem]{Proposition}

\theoremstyle{definition}
\newtheorem{definition}[theorem]{Definition}
\newtheorem{example}[theorem]{Example}
\newtheorem{remark}[theorem]{Remark}

\newcommand{\comp}{\makebox[7pt]{\raisebox{1pt}{\tiny $\circ$}}}
\newcommand{\R}{\mathbb{R}} 
\newcommand{\F}{\mathbb{F}}

\newcommand{\Leg}{\mathbb{F}L}
\newcommand{\hor}{\operatorname{hor}}
\newcommand{\ver}{\operatorname{ver}}
\newcommand{\id}{{\operatorname{id}}}
\newcommand{\pr}{\mathrm{pr}}

\def\od#1#2{\dfrac{d#1}{d#2}}
\def\pd#1#2{\dfrac{\partial #1}{\partial #2}}
\def\tpd#1#2{\partial #1/\partial #2}

\def\parentheses#1{\left(#1\right)}
\def\brackets#1{\left[#1\right]}
\def\braces#1{\left\{#1\right\}}

\def\Span{\operatorname{span}} 

\def\DS{\displaystyle}
\def\defeq{\mathrel{\mathop:}=}
\def\eqdef{=\mathrel{\mathop:}}
\def\setdef#1#2{ \left\{ #1 \ |\ #2 \right\} }
\def\ip#1#2{\left\langle#1,#2\right\rangle}
\def\hl{\operatorname{hl}}
\def\eps{\varepsilon}

\begin{document}


\title[H--J Theory for Lagrange--Dirac Systems]
{Hamilton--Jacobi Theory for Degenerate Lagrangian Systems with Holonomic and Nonholonomic Constraints}

\author{Melvin Leok}
\email{mleok@math.ucsd.edu}

\author{Tomoki Ohsawa}
\email{tohsawa@ucsd.edu}
\affiliation{Department of Mathematics, University of California, San Diego, 9500 Gilman Drive, La Jolla, California, USA}

\author{Diana Sosa}
\email{dnsosa@ull.es}
\affiliation{Departamento de Econom\'{\i}a Aplicada y Unidad Asociada ULL-CSIC
  Geo\-me\-tr\'{\i}a Diferencial y Mec\'anica Geom\'etrica, Facultad de
  CC. EE. y Empresariales, Universidad de La Laguna, La Laguna,
  Tenerife, Canary Islands, Spain}

\date{\today}

\begin{abstract}
  We extend Hamilton--Jacobi theory to Lagrange--Dirac (or implicit Lagrangian) systems, a generalized formulation of Lagrangian mechanics that can incorporate degenerate Lagrangians as well as holonomic and nonholonomic constraints.
  We refer to the generalized Hamilton--Jacobi equation as the {\em Dirac--Hamilton--Jacobi equation}.
  For non-degenerate Lagrangian systems with nonholonomic constraints, the theory specializes to the recently developed nonholonomic Hamilton--Jacobi theory.
  We are particularly interested in applications to a certain class of degenerate nonholonomic Lagrangian systems with symmetries, which we refer to as {\em weakly degenerate Chaplygin systems}, that arise as simplified models of nonholonomic mechanical systems; these systems are shown to reduce to non-degenerate almost Hamiltonian systems, i.e., generalized Hamiltonian systems defined with non-closed two-forms.
  Accordingly, the Dirac--Hamilton--Jacobi equation reduces to a variant of the nonholonomic Hamilton--Jacobi equation associated with the reduced system.
  We illustrate through a few examples how the Dirac--Hamilton--Jacobi equation can be used to exactly integrate the equations of motion.
\end{abstract}

\pacs{02.40.Yy, 45.20.Jj}
\keywords{Hamilton--Jacobi theory, Lagrange--Dirac systems, Dirac structures, degenerate or singular Lagrangian systems, nonholonomic constraints.}

\maketitle

\section{Introduction}
\subsection{Degenerate Lagrangian Systems and Lagrange--Dirac Systems}
Degenerate Lagrangian systems are the motivation behind the work
of \citet{Di1950,Di1958a,Di1964} on constrained systems, where
degeneracy of Lagrangians imposes constraints on the phase space
variables. The theory gives a prescription for writing such
systems as Hamiltonian systems, and is used extensively for gauge
systems and their quantization~(see, e.g., \citet{HeTe1992}).

Dirac's theory of constraints was geometrized by
\citet{GoNeHi1978} (see also \citet{GoNe1979a,GoNe1979b,GoNe1980}
and \citet{Ku1969}) to yield a constraint algorithm to identify
the solvability condition for presymplectic systems and also to
establish the equivalence between Lagrangian and Hamiltonian
descriptions of degenerate Lagrangian systems. The algorithm is
extended by \citet{LeMa1997a} to degenerate Lagrangian systems
with nonholonomic constraints.

On the other hand, Lagrange--Dirac (or implicit Lagrangian)
systems of \citet{YoMa2006a,YoMa2006b} provide a rather direct way of
describing degenerate Lagrangian systems that do not explicitly
involve constraint algorithms. Moreover, the Lagrange--Dirac
formulation can address more general constraints, particularly
nonholonomic constraints, by directly encoding them in terms of
Dirac structures, as opposed to symplectic or Poisson structures.

\subsection{Hamilton--Jacobi Theory for Constrained Degenerate Lagrangian Systems}
The goal of this paper is to generalize Hamilton--Jacobi theory to
Lagrange--Dirac systems. The challenge in doing so is to
generalize the theory to simultaneously address degeneracy and
nonholonomic constraints. For degenerate Lagrangian systems, some
work has been done, built on Dirac's theory of constraints, on
extending Hamilton--Jacobi theory (see, e.g., \citet{HeTe1992}[Section~5.4] and \citet{RoSc2003}) as well as
from the geometric point of view by \citet{CaGrMaMaMuRo2006}. For
nonholonomic systems, \citet{IgLeMa2008} generalized the geometric
Hamilton--Jacobi theorem (see Theorem~5.2.4 of \citet{AbMa1978}) to
nonholonomic systems, which has been studied further by
\citet{LeMaMa2010}, \citet{OhBl2009}, \citet{CaGrMaMaMuRo2010},
and \citet{OhFeBlZe2011}. However, to the authors' knowledge, no
work has been done that can deal with both degeneracy and
nonholonomic constraints.

\subsection{Applications to Degenerate Lagrangian Systems with Nonholonomic Constraints}
We are particularly interested in applications to degenerate
Lagrangian systems with nonholonomic constraints. Such systems
arise regularly, in practice, as model reductions of multiscale
systems: For example, consider a nonholonomic mechanical system
consisting of rigid bodies, some of which are significantly
lighter than the rest. Then, one can make an assumption that the
light parts are massless for the sake of simplicity; this often
results in a degenerate Lagrangian. While na\"\i vely making a
massless approximation usually leads to unphysical
results\footnote{This setting usually gives a singular
perturbation problem with the small mass being its parameter. For
example, for the Lagrangian $L_{\eps} = \eps \dot{x}^{2}/2 -
x^{2}/2$, the Euler--Lagrange equation gives $\eps \ddot{x} + x = 0$;
the solution corresponding to the massless Lagrangian $L_{0}$
deviates significantly from the original solution.}, a certain
class of nonholonomic systems seem to allow massless
approximations without such inconsistencies. See, for example, the
modelling of a bicycle in \citet{Ge1994} and \citet{GeMa1995} (see
also \citet{KoMa1997c} and Example~\ref{ex:Bicycle} of the present paper).

\subsection{Outline}
We first briefly review Dirac structures and Lagrange--Dirac
systems in Section~\ref{sec:LDS}. Section~\ref{sec:WDCS}
introduces a class of degenerate nonholonomic Lagrangian systems
with symmetries that reduce to non-degenerate Lagrangian systems
after symmetry reduction; we call them {\em weakly degenerate
Chaplygin systems}. Section~\ref{sec:DiracHJ} gives
Hamilton--Jacobi theory for Lagrange--Dirac systems, defining the
{\em Dirac--Hamilton--Jacobi equation}, and shows applications to
degenerate Lagrangian systems with holonomic and nonholonomic
constraints. We then apply the theory to weakly degenerate
Chaplygin systems in Section~\ref{sec:DiracHJ-WDCS}; we derive a
formula that relates solutions of the Dirac--Hamilton--Jacobi
equations with those of the nonholonomic Hamilton--Jacobi equation
for the reduced weakly degenerate Chaplygin systems.
Appendix~\ref{sec:RedWDCS} discusses reduction of weakly
degenerate Chaplygin systems by a symmetry reduction of the associated Dirac structure.

\section{Lagrange--Dirac Systems}
\label{sec:LDS} Lagrange--Dirac (or implicit Lagrangian) systems
are a generalization of Lagrangian mechanics to systems with
(possibly) degenerate Lagrangians and constraints. Given a
configuration manifold $Q$, a Lagrange--Dirac system is defined
using a generalized Dirac structure on $T^{*}Q$, or more precisely
a subbundle $D$ of the Whitney sum $TT^{*}Q \oplus T^{*}T^{*}Q$.

\subsection{Dirac Structures}
Let us first recall the definition of a (generalized) Dirac
structure on a manifold $M$.
Let $M$ be a manifold.
Given a subbundle $D \subset TM \oplus T^{*}M$, the subbundle $D^{\perp} \subset TM \oplus T^{*}M$ is defined as follows:
\begin{equation*}
  D^{\perp} \defeq \setdef{ (X, \alpha) \in TM \oplus T^{*}M }{ \ip{\alpha'}{X} + \ip{\alpha}{X'} = 0 \text{ for any $(X',\alpha') \in D$} }.
\end{equation*}
\begin{definition}
  A subbundle $D \subset TM \oplus T^{*}M$ is called a {\em generalized Dirac structure} if $D^{\perp} = D$.
\end{definition}
Note that the notion of Dirac structures, originally introduced in
\citet{co1990a}, further satisfies an integrability condition, which
we have omitted as it is not compatible with our interest in
nonintegrable (nonholonomic) constraints. Hereafter, we refer to
generalized Dirac structures as simply ``Dirac structures.''

\subsection{Induced Dirac Structures}
Here we consider the induced Dirac structure $D_{\Delta_{Q}}
\subset TT^{*}Q \oplus T^{*}T^{*}Q$ introduced in
\citet{YoMa2006a}. See \citet{DaSc1998} for more general Dirac
structures, \citet{BlCr1997} and \citet{Sc1998} for those defined
by Kirchhoff current and voltage laws, and \citet{Sc2006} for
applications of Dirac structures to interconnected systems.

Let $Q$ be a smooth manifold, $\Delta_Q\subset TQ$ a regular
distribution on $Q$, and $\Omega$ the canonical symplectic
two-form on $T^*Q$. Denote by $\Delta_Q^\circ$ the annihilator of
$\Delta_Q$ and by $\Omega^\flat:TT^*Q\to T^*T^*Q$ the flat map induced by $\Omega$.
The distribution $\Delta_Q\subset TQ$ may be lifted to
the distribution $\Delta_{T^*Q}$ on $T^*Q$ defined as
\begin{equation*}
\Delta_{T^*Q} \defeq (T\pi_Q)^{-1}(\Delta_Q)\subset TT^*Q,
\end{equation*}
where $\pi_Q:T^*Q\to Q$ is the canonical projection and
$T\pi_Q:TT^*Q\to TQ$ is its tangent map. Denote its annihilator by
$\Delta^\circ_{T^*Q}\subset T^*T^*Q$.

\begin{definition}[\citet{YoMa2006a,YoMa2006b}; see also \citet{DaSc1998}]
  The {\em induced (generalized) Dirac structure} $D_{\Delta_Q}$ on $T^*Q$
  is defined, for each $z\in T^* Q$, as
  \begin{equation*}
    D_{\Delta_{Q}}(z) \defeq \setdef{
      (v_z,\alpha_z) \in T_z T^*Q \oplus T^*_z T^*Q
    }{
      v_z\in\Delta_{T^*Q}(z),\;
      \alpha_z-\Omega^{\flat}(z)(v_z) \in \Delta^\circ_{T^*Q}(z)
    }.
  \end{equation*}
\end{definition}
If we choose local coordinates $q=(q^i)$ on an open subset $U$ of
$Q$ and denote by $(q,\dot{q})=(q^i,\dot{q}^i)$ (respectively,
$(q,p)=(q^i,p_i)$), the corresponding local coordinates on $TQ$
(respectively, $T^*Q$), then a local representation for the Dirac
structure is given by
\begin{multline*}
  D_{\Delta_Q}(q,p) = \left\{
    ((q,p,\dot{q},\dot{p}),(q,p,\alpha_q,\alpha_p)) \in T_{(q,p)}T^*Q \oplus T^{*}_{(q,p)} T^*Q
    \ |\
    \right.
    \\
    \left.
    \dot{q}\in\Delta_Q(q), \;
    \alpha_p=\dot{q}, \;
    \alpha_q+\dot{p}\in\Delta_Q^\circ(q)
    \right\}.
\end{multline*}

\subsection{Lagrange--Dirac Systems}
To define a Lagrange--Dirac system, it is necessary to introduce the Dirac differential of a Lagrangian function.
Following \citet{YoMa2006a}, let us first introduce the following maps, originally due to \citet{Tu1976a,Tu1976b}, between the iterated tangent and cotangent bundles.
\begin{equation}
  \label{eq:gamma_Q}
  \begin{array}{c}
    \xymatrix@!0@R=1.5in@C=0.95in{
      T^{*}TQ \ar@/^{1.5pc}/[rr]^{\gamma_{Q}}  &
      TT^{*}Q \ar[l]^{\kappa_{Q}} \ar[r]_{\Omega^{\flat}} &
      T^{*}T^{*}Q
    }
    \qquad
    \xymatrix@!0@R=1.5in@C=1.15in{
      (q, \delta q, \delta p, p) \ar@{|->}@/^{1.5pc}/[rr] &
      (q, p, \delta q, \delta p) \ar@{|->}[l] \ar@{|->}[r] &
      (q, p, -\delta p, \delta q) 
    }
  \end{array}
\end{equation}
Let $L:TQ\to \R$ be a Lagrangian function and let
$\gamma_Q:T^*TQ\to T^*T^*Q$ be the diffeomorphism defined as
$\gamma_Q \defeq \Omega^\flat \comp \kappa_Q^{-1}$ (see \eqref{eq:gamma_Q}). Then, the Dirac differential of $L$ is
the map $\mathfrak{D}L:TQ\to T^*T^*Q$ given by
\begin{equation*}\mathfrak{D}L=\gamma_Q \comp dL.\end{equation*}
In local coordinates,
\begin{equation*}
\mathfrak{D}L(q,v)=\left(q,\displaystyle\frac{\partial L}{\partial
v},-\frac{\partial L}{\partial q},v\right).
\end{equation*}

\begin{definition}
  Let $L:TQ\to \R$ be a Lagrangian (possibly degenerate) and $\Delta_Q\subset
  TQ$ be a given regular constraint distribution on the
  configuration manifold $Q$.
  Let
  \begin{equation*}
    P \defeq {\mathbb F}L(\Delta_Q) \subset T^*Q
  \end{equation*}
  be the image of $\Delta_Q$ by the Legendre transformation
  and $X$ be a (partial) vector field on $T^*Q$ defined at points of $P$.
  Then, a {\em Lagrange--Dirac system} is the triple
  $(L,\Delta_Q,X)$ that satisfies, for each point $z \in P\subset
  T^*Q$,
  \begin{equation}
    \label{eq:LDS}
    (X(z),\mathfrak{D}L(u))\in D_{\Delta_Q}(z),
  \end{equation}
  where $u\in \Delta_Q$ such that ${\mathbb F}L(u)=z$.
  In local coordinates, Eq.~\eqref{eq:LDS} is written as
  \begin{equation}
    \label{eq:LDEq}
    p = \frac{\partial L}{\partial v}(q,v),
    \qquad
    \dot{q} \in \Delta_Q(q),
    \qquad
    \dot{q}=v,
    \qquad
    \dot{p} - \frac{\partial L}{\partial q}(q,v)\in\Delta_Q^\circ(q),
  \end{equation}
  which we call the {\em Lagrange--Dirac equations}.
\end{definition}
We note that the idea of applying implicit differential equations to nonholonomic systems is found in an earlier work by \citet{IbLeMaMa1996}; see also \citet{GrGr2008} for a generalization to vector bundles with algebroid structures.

\begin{definition}
  A {\em solution curve} of a Lagrange--Dirac system $(L,\Delta_Q,X)$ is an integral curve $(q(t),p(t))$,
  $t_1\leq t \leq t_2$, of $X$ in $P\subset T^*Q$.
\end{definition}

\subsection{Lagrange--Dirac Systems on the Pontryagin Bundle {\boldmath $TQ \oplus T^{*}Q$}}
\label{ssec:LDSonPontryaginBundle}
We may also define a Lagrange--Dirac system on $TQ\oplus T^*Q$ as well.
We will use the submanifold ${\mathcal K}$ of the Pontryagin bundle introduced in Yoshimura and Marsden
\cite{YoMa2006a} and the (partial) vector field $\tilde{X}$ on
$TQ\oplus T^*Q$, associated with a (partial) vector field $X$ on
$T^*Q$, defined in Yoshimura and Marsden \cite{YoMa2006b}. Let us
recall the definition of these two objects.

Given a Lagrangian $L:TQ\to \R$, the {\em generalized energy},
$\mathcal{E}:TQ\oplus T^*Q\to\R$, is given by
\begin{equation*}
\mathcal{E}(q,v,p)=p\cdot v-L(q,v).
\end{equation*}
The submanifold ${\mathcal K}$ is defined as the set of stationary points of 
$\mathcal{E}(q,v,p)$ with respect to $v$, with $v\in\Delta_Q(q)$.
So, ${\mathcal K}$ is represented by
\begin{equation}
  \label{eq:mathcalK}
  {\mathcal K} = \setdef{
    (q,v,p)\in TQ\oplus T^*Q
  }{
    v\in\Delta_Q(q),\;
    p = \frac{\partial L}{\partial v}(q,v)
  }
\end{equation}
This submanifold can also be described as the graph of the
Legendre transformation restricted to the constraint distribution
$\Delta_Q$. We can also obtain the submanifold ${\mathcal K}$ as
follows.
Let $\pr_{TQ}: TQ \oplus T^{*}Q \to TQ$ be the projection to the first factor and $\pi_{TQ}:T^*TQ\to TQ$ be the cotangent bundle projection.
Consider the map $\rho_{T^*TQ}:T^*TQ\to TQ\oplus T^*Q$ (see \citet{YoMa2006a}[Section~4.10]) which has the property that $\pr_{TQ}\comp\rho_{T^*TQ}=\pi_{TQ}$; this map is defined intrinsically to be the direct sum of $\pi_{TQ}:T^*TQ\to TQ$ and $\tau_{T^*Q} \comp \kappa_Q^{-1}:T^*TQ\to T^*Q$ (see \citet{YoMa2006a}[Section 4.10]), where $\tau_{T^{*}Q}: TT^{*}Q \to T^{*}Q$ is the tangent bundle projection.
Then, we can consider the map
\begin{equation*}
  \rho_{T^*TQ} \comp dL: TQ\to TQ\oplus T^*Q,
\end{equation*}
whose local expression is
\begin{equation*}
  \rho_{T^*TQ} \comp dL(q,v) =\left(q,v,\displaystyle\frac{\partial
      L}{\partial v}(q,v)\right).
\end{equation*}
Therefore, we have
\begin{equation*}
  \mathcal{K} =
  \rho_{T^*TQ} \comp dL(\Delta_Q).
\end{equation*}

Now, given a (partial) vector field $X$ on $T^*Q$ defined at
points of $P$, one can construct a (partial) vector field
$\tilde{X}$ on $TQ\oplus T^*Q$ defined at points of ${\mathcal K}$
as follows (see \citet{YoMa2006b}[Section 3.8]). For
$(q,v,p)\in{\mathcal K}$, $\tilde{X}(q,v,p)$ is tangent to a curve
$(q(t),v(t),p(t))$ in $TQ\oplus T^*Q$ such that
$(q(0),v(0),p(0))=(q,v,p)$ and $X(q,p)$ is tangent to the curve $(q(t),p(t))$ in $T^*Q$.  This (partial) vector field $\tilde{X}$ is not
unique; however it has the property that, for each
$x\in{\mathcal K}\subset TQ\oplus T^*Q$,
\begin{equation*}
  T \pr_{T^*Q}(\tilde{X}(x))=X(\pr_{T^*Q}(x)),
\end{equation*}
where $\pr_{T^{*}Q}: TQ \oplus T^{*}Q \to T^{*}Q$ is the projection to the second factor.

On the other hand, from the distribution $\Delta_Q$ on $Q$, we can
define a distribution $\Delta_{TQ\oplus T^*Q}$ on $TQ\oplus T^*Q$
by
\begin{equation*}
  \Delta_{TQ\oplus T^*Q}=(T\pr_Q)^{-1}(\Delta_Q),
\end{equation*}
where $\pr_Q:TQ\oplus T^*Q\to Q$. Note that $\Delta_{TQ\oplus
T^*Q}=(T\pr_{T^*Q})^{-1}(\Delta_{T^*Q})$, since $\pr_Q=\pi_Q\comp
\pr_{T^*Q}$. Then, as $\pr^*_{T^*Q}\Omega$ is a skew-symmetric
two-form on $TQ\oplus T^*Q$, we can consider the following induced
(generalized) Dirac structure on $TQ\oplus T^*Q$:
\begin{multline*}
  D_{TQ\oplus T^*Q}(x) \defeq
  \Bigl\{(\tilde{v}_x,\tilde{\alpha}_x)\in T_x(TQ\oplus T^*Q) \oplus T^*_x(TQ\oplus T^*Q) \ |\  \\
  \tilde{v}_x\in \Delta_{TQ\oplus T^*Q}(x),\;
  \tilde{\alpha}_x-(\pr^*_{T^*Q}\Omega)^\flat (x)(\tilde{v}_x) \in \Delta_{TQ\oplus T^*Q}^{\circ}(x)\Bigr\},
\end{multline*}
for $x \in TQ\oplus T^*Q$. A local representation for the Dirac
structure $D_{TQ\oplus T^*Q}$ is
\begin{multline*}
  D_{TQ\oplus T^*Q}(q,v,p) = 
  \Bigl\{((q,v,p,\dot{q},\dot{v},\dot{p}),(q,v,p,\tilde{\alpha}_q,\tilde{\alpha}_v,\tilde{\alpha}_p)) \ |\ \\
   \dot{q}\in\Delta_Q(q), \;
  \tilde{\alpha}_p=\dot{q}, \;
  \tilde{\alpha}_v=0, \;
  \tilde{\alpha}_q+\dot{p}\in\Delta^\circ_Q(q)\Bigr\}.
\end{multline*}
Then, we have the following result.
\begin{theorem}
  For every $u\in \Delta_Q$, define $z \defeq {\mathbb F}L(u)\in P$ and $x \defeq \rho_{T^*TQ} \comp dL(u) \in {\mathcal K}$ so that $\pr_{T^*Q}(x) = z$.
  Then, we have
  \begin{equation*}
    (X(z),\mathfrak{D}L(u))\in D_{\Delta_Q}(z)
    \iff
    (\tilde{X}(x),d\mathcal{E}(x))\in D_{TQ\oplus T^*Q}(x).
  \end{equation*}
\end{theorem}

\begin{proof}It is not difficult to prove that the condition $(\tilde{X}(x),d\mathcal{E}(x))\in D_{TQ\oplus T^*Q}(x)$ locally
reads
\begin{equation*}
  p = \frac{\partial L}{\partial v}(q,v),
  \qquad
  \dot{q} \in \Delta_Q(q),
  \qquad
  \dot{q}=v,
  \qquad
  \dot{p} - \frac{\partial L}{\partial q}(q,v)\in\Delta_Q^\circ(q),
\end{equation*}
that is, the Lagrange--Dirac equations~\eqref{eq:LDEq}; thus we have the equivalence.
\end{proof}

As a consequence, we obtain the following result which was
obtained by Yoshimura and Marsden (see Theorem 3.8 in \citet{YoMa2006a}).

\begin{corollary} If $(q(t),p(t))={\mathbb F}L(q(t),v(t))$, $t_1\leq t\leq
t_2$, is an integral curve of the vector field $X$ on $P$, then
  $\rho_{T^*TQ} \comp dL(q(t),v(t))$ is an integral curve of $\tilde{X}$ on ${\mathcal K}$.  Conversely, if $(q(t),v(t),p(t))$, $t_1\leq t\leq t_2$,
  is an integral curve of $\tilde{X}$ on ${\mathcal K}$, then $\pr_{T^*Q}(q(t),v(t),p(t))$ is an integral curve of $X$.
\end{corollary}

Therefore, a Lagrange--Dirac system on the Pontryagin bundle is
given by a triple $(\mathcal{E},\mathcal{K},\tilde{X})$ satisfying
the condition
\begin{equation*}
  (\tilde{X}(x),d\mathcal{E}(x))\in D_{TQ\oplus T^*Q}(x),
\end{equation*}
 for all $x\in\mathcal{K}$.

\section{Degenerate Lagrangian Systems with Nonholonomic Constraints}
\label{sec:WDCS} If one accurately models a mechanical system,
then one usually obtains a non-degenerate Lagrangian, since the
kinetic energy of the system is usually written as a
positive-definite quadratic form in their velocity components.
However, for a complex mechanical system consisting of many moving
parts, one can often ignore the masses and/or moments of inertia
of relatively light parts of the system in order to simplify the
analysis. This turns out to be an effective way of modeling
complex systems; for example, one usually models the strings of a
puppet as massless moving parts (see, e.g., \citet{JoMu2007} and
\citet{MuEg2007}). With such an approximation, the Lagrangian
often turns out to be degenerate, and thus the Euler--Lagrange or
Lagrange--d'Alembert equations do not give the dynamics of the
massless parts directly; instead, it is determined by mechanical
constraints. In other words, the system may be considered as a
hybrid of dynamics and kinematics.

We are particularly interested in systems with degenerate
Lagrangians and {\em nonholonomic} constraints, because they
possess the two very features that Lagrange--Dirac systems can
(and are designed to) incorporate but the standard Lagrangian or
Hamiltonian formulation cannot.

In this section, we introduce a class of mechanical systems with
degenerate Lagrangians and nonholonomic constraints with symmetry
that yield non-degenerate almost Hamiltonian systems\footnote{An {\em almost Hamiltonian system} is a generalized Hamiltonian system defined with a non-degenerate but non-closed two-form as opposed to a symplectic form (which is closed by definition)~\cite{BaSn1993, HoGa2009}.}
on the reduced
space when symmetry reduction is performed.

\subsection{Chaplygin Systems}
\label{ssec:ChaplyginSystems}
Let us start from the following definition of a well-known class
of nonholonomic systems:
\begin{definition}[Chaplygin Systems; see, e.g., \citet{Ko1992}, {\citet{Co2004}[Chapters~4 \& 5]} and \citet{HoGa2009}]
  \label{def:ChaplyginSystems}
  A nonholonomic system with Lagrangian $L$ and distribution $\Delta_{Q}$ is called a {\em Chaplygin system} if there exists a Lie group $G$ with a free and proper action on $Q$, i.e., $\Phi: G \times Q \to Q$ or $\Phi_{g}: Q \to Q$ for any $g \in G$, such that
  \begin{enumerate}[(i)]
  \item the Lagrangian $L$ and the distribution $\Delta_{Q}$ are invariant under the tangent lift of the $G$-action, i.e., $L \comp T\Phi_{g} = L$ and $T\Phi_{g}(\Delta_{Q}(q)) = \Delta_{Q}(g q)$;
  \item for each $q \in Q$, the tangent space $T_{q}Q$ is the direct sum of the constraint distribution and the tangent space to the orbit of the group action, i.e.,
    \begin{equation*}
      T_{q}Q = \Delta_{Q}(q) \oplus T_{q}\mathcal{O}_{q},
    \end{equation*}
    where $\mathcal{O}_{q}$ is the orbit through $q$ of the $G$-action on $Q$, i.e.,
    \begin{equation*}
      \mathcal{O}_{q} \defeq \setdef{ \Phi_{g}(q) \in Q }{ g \in G }.
    \end{equation*}
  \end{enumerate}
\end{definition}

This setup gives rise to the principal bundle
\begin{equation*}
  \pi: Q \to Q/G \eqdef \bar{Q}
\end{equation*}
and the connection
\begin{equation}
  \label{eq:mathcalA}
  \mathcal{A}: TQ \to \mathfrak{g},
\end{equation}
with $\mathfrak{g}$ being the Lie algebra of $G$ such that
$\ker\mathcal{A} = \Delta_{Q}$, i.e., the horizontal space of
$\mathcal{A}$ is $\Delta_Q$. Furthermore, for any $q \in Q$ and
$\bar{q} \defeq \pi(q) \in \bar{Q}$, the map
$T_{q}\pi|_{\Delta_{Q}(q)}: \Delta_{Q}(q) \to T_{\bar{q}}\bar{Q}$
is a linear isomorphism, and hence we have the horizontal lift
\begin{equation*}
  \hl^{\Delta}_{q}: T_{\bar{q}}\bar{Q} \to \Delta_{Q}(q);
  \quad
  v_{\bar{q}} \mapsto (T_{q}\pi|_{\Delta_{Q}(q)})^{-1}(v_{\bar{q}}).
\end{equation*}
We will occasionally use the following shorthand notation for
horizontal lifts:
\begin{equation*}
  v^{\rm h}_{q} \defeq \hl^{\Delta}_{q}(v_{\bar{q}}).
\end{equation*}
Then, any vector $W_{q} \in T_{q}Q$ can be decomposed into
the horizontal and vertical parts as follows:
\begin{equation*}
  W_{q} = \hor(W_{q}) + \ver(W_{q}),
\end{equation*}
with
\begin{equation*}
  \hor(W_{q}) = \hl^{\Delta}_{q}(w_{\bar{q}}),
  \qquad
  \ver(W_{q}) = (\mathcal{A}_{q}(W_{q}))_{Q}(q),
\end{equation*}
where $w_{\bar{q}} \defeq T_{q}\pi(W_{q})$ and $\xi_{Q} \in
\mathfrak{X}(Q)$ is the infinitesimal generator of $\xi \in
\mathfrak{g}$.

Suppose that the Lagrangian $L: TQ \to \R$ is of the form
\begin{equation}
  \label{eq:SimpleLagrangian}
  L(v_{q}) = \frac{1}{2}g_{q}(v_{q}, v_{q}) - V(q),
\end{equation}
where $g$ is a possibly degenerate metric on $Q$.
We may then define the reduced Lagrangian
\begin{equation*}
  \bar{L} \defeq L \comp \hl^{\Delta},
\end{equation*}
or more explicitly,
\begin{equation*}
  \bar{L}: T\bar{Q} \to \R;
  \quad
  v_{\bar{q}} \mapsto \frac{1}{2} \bar{g}_{\bar{q}}(v_{\bar{q}}, v_{\bar{q}}) - \bar{V}(\bar{q}),
\end{equation*}
where $\bar{g}$ is the metric on the reduced space $\bar{Q}$
induced by $g$ as follows:
\begin{equation*}
  \bar{g}_{\bar{q}}(v_{\bar{q}}, w_{\bar{q}})
  \defeq g_{q}\parentheses{ \hl^{\Delta}_{q}(v_{\bar{q}}), \hl^{\Delta}_{q}(w_{\bar{q}}) }
  = g_{q}( v^{\rm h}_{q}, w^{\rm h}_{q}),
\end{equation*}
and the reduced potential $\bar{V}: \bar{Q} \to \R$ is defined
such that $V = \bar{V} \comp \pi$.

\subsection{Weakly Degenerate Chaplygin Systems}
\label{ssec:WeaklyDegenerateChaplyginSystems}
The following special class of Chaplygin systems is of particular
interest in this paper:
\begin{definition}[Weakly Degenerate Chaplygin Systems]
  A Chaplygin system is said to be {\em weakly degenerate} if the Lagrangian $L: TQ \to \R$ is degenerate but the reduced Lagrangian $\bar{L}: T\bar{Q} \to \R$ is non-degenerate; more precisely, the metric $g$ is degenerate on $TQ$ but positive-definite (hence non-degenerate) when restricted to $\Delta_{Q} \subset TQ$, i.e., the triple $(Q, \Delta_{Q}, g)$ defines a sub-Riemannian manifold~(see, e.g., \citet{Mo2002}), and the induced metric $\bar{g}$ on $\bar{Q}$ is positive-definite and hence Riemannian.
\end{definition}

\begin{remark}
  This is a mathematical description of the hybrid of dynamics and kinematics mentioned above:
  The dynamics is essentially dropped to the reduced configuration manifold $\bar{Q} \defeq Q/G$, and the rest is reconstructed by the horizontal lift $\hl^{\Delta}$, which is the kinematic part defined by the (nonholonomic) constraints.
\end{remark}

\begin{remark}
  Note that the positive-definiteness of the metric $g$ on $\Delta_{Q}$ guarantees that a weakly degenerate Chaplygin system is regular in the sense of \citet{LeMa1996e} (see Proposition~II.4 therein and also \citet{LeMaMa1997}).
\end{remark}

We will look into the geometry associated with weakly degenerate
Chaplygin systems in Section~\ref{ssec:GeometryOfWDCS}.

\begin{example}[Simplified Roller Racer; see \citet{Ts1995} and \citet{KrTs2001} and {\citet{Bl2003}[Section~1.10]}]
  \label{ex:RollerRacer}
  The roller racer, shown in Fig.~\ref{fig:RollerRacer}, consists of two (main and second) planar coupled rigid bodies, each of which has a pair of wheels attached at its center of mass.
  We assume that {\em the mass of the second body is negligible, and hence so are its kinetic and rotational energies}\footnote{We note that, in the original model~\cite{Ts1995,KrTs2001}, only the kinetic energy of the second body is ignored, and its rotational energy is taken into account; one obtains a {\em non-degenerate} Lagrangian with such an approximation.}.
  Let $(x,y)$ be the coordinates of the center of mass of the main body, $\theta$ the angle of the line passing through the center of mass measured from the $x$-axis, $\phi$ the angle between the two bodies; $d_{1}$ and $d_{2}$ are the distances from centers of mass to the joint, $m_{1}$ and $I_{1}$ the mass and inertia of the main body.

  The configuration space is $Q = SE(2) \times \mathbb{S}^{1} = \{ (x, y, \theta, \phi) \}$, and the Lagrangian $L: TQ \to \R$ is given by
  \begin{equation*}
    L = \frac{1}{2}m_{1} (\dot{x}^{2} + \dot{y}^{2}) + \frac{1}{2}I_{1} \dot{\theta}^{2},
  \end{equation*}
  which is degenerate because of the massless approximation of the second body.

  \begin{figure}[htbp]
    \centering
    \includegraphics[width=.6\linewidth]{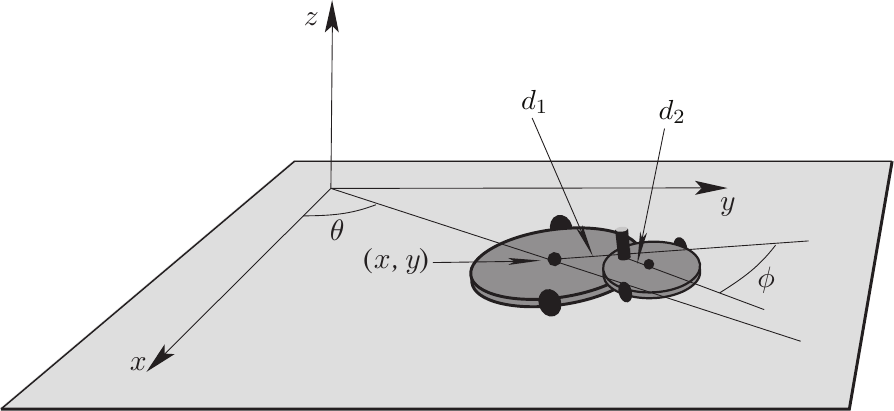}
    \caption{Roller Racer (taken from \citet{Bl2003} with permission from the author). The mass of the second body is assumed to be negligible.}
    \label{fig:RollerRacer}
  \end{figure}
  The constraints are given by
  \begin{equation}
    \label{eq:Constraints-RollerRacer}
    \dot{x} = \cos\theta\csc\phi \brackets{ (d_{1}\cos\phi + d_{2})\dot{\theta} + d_{2} \dot{\phi} },
    \qquad
    \dot{y} = \sin\theta\csc\phi \brackets{ (d_{1}\cos\phi + d_{2})\dot{\theta} + d_{2} \dot{\phi} }.
  \end{equation}
  Defining the constraint one-forms
  \begin{equation}
    \label{eq:omegas-RollerRacer}
    \omega^{1} \defeq dx - \cos\theta\csc\phi[ (d_{1}\cos\phi + d_{2})d\theta + d_{2}\,d\phi ],
    \qquad
    \omega^{2} \defeq dy - \sin\theta\csc\phi[ (d_{1}\cos\phi + d_{2})d\theta + d_{2}\,d\phi ],
  \end{equation}
  we can write the constraint distribution $\Delta_{Q} \subset TQ$ as
  \begin{equation*}
    \Delta_{Q} = \setdef{ \dot{q} = (\dot{x}, \dot{y}, \dot{\theta}, \dot{\phi}) \in TQ }{ \omega^{a}(\dot{q}) = 0,\, a = 1,2 }.
  \end{equation*}

  The Lagrange--Dirac equations~\eqref{eq:LDEq} give
  \begin{equation}
    \label{eq:LDEq-RollerRacer}
    \begin{array}{c}
      \DS
      p_{x} = m_{1} v_{x},
      \qquad
      p_{y} = m_{1} v_{y},
      \qquad
      p_{\theta} = I_{1} v_{\theta},
      \qquad
      p_{\phi} = 0,
      \medskip\\
      \DS
      \dot{x} = \cos\theta\csc\phi \brackets{ (d_{1}\cos\phi + d_{2})\dot{\theta} + d_{2} \dot{\phi} },
      \qquad
      \dot{y} = \sin\theta\csc\phi \brackets{ (d_{1}\cos\phi + d_{2})\dot{\theta} + d_{2} \dot{\phi} },
      \medskip\\
      \DS
      \dot{x} = v_{x},
      \qquad
      \dot{y} = v_{y},
      \qquad
      \dot{\theta} = v_{\theta},
      \qquad
      \dot{\phi} = v_{\phi},
      \medskip\\
      \DS
      \dot{p}_{x} = \lambda \sin\theta,
      \qquad
      \dot{p}_{y} = -\lambda \cos\theta,
      \qquad
      \dot{p}_{\theta} = 0,
      \qquad
      \dot{p}_{\phi} = 0,
    \end{array}
  \end{equation}
  where $\lambda$ is the Lagrange multiplier.

  Let $G = \R^{2}$ and consider the action of $G$ on $Q$ by translations on the $x$-$y$ plane, i.e.,
  \begin{equation*}
    G \times Q \to Q;
    \quad
    \parentheses{(a, b), (x, y, \theta, \phi)} \mapsto (x + a, y + b, \theta, \phi).
  \end{equation*}
  Then, the tangent space to the group orbit is given by
  \begin{equation*}
    T_{q}\mathcal{O}(q) = \Span\braces{ \pd{}{x}, \pd{}{y} },
  \end{equation*}
  with $q = (x, y, \theta, \phi)$.
  It is easy to check that this defines a Chaplygin system in the sense of Definition~\ref{def:ChaplyginSystems}.
  The quotient space is $\bar{Q} \defeq Q/G = \{(\theta, \phi)\}$, and the horizontal lift $\hl^{\Delta}$ is
  \begin{equation*}
    \hl^{\Delta}_{q}(\dot{\theta}, \dot{\phi}) = \parentheses{
      \cos\theta\csc\phi \brackets{ (d_{1}\cos\phi + d_{2})\dot{\theta} + d_{2} \dot{\phi} },
      \sin\theta\csc\phi \brackets{ (d_{1}\cos\phi + d_{2})\dot{\theta} + d_{2} \dot{\phi} },
      \dot{\theta}, \dot{\phi}
    }.
  \end{equation*}
  Hence, the reduced Lagrangian $\bar{L}: T\bar{Q} \to \R$ is given by
  \begin{equation}
    \label{eq:barL-RollerRacer}
    \bar{L} = \frac{1}{2} m_{1} \parentheses{ d_{1} \dot{\theta} \cos\phi + d_{2}(\dot{\theta} + \dot{\phi}) }^{2} \csc^{2}\phi + \frac{1}{2}I_{1} \dot{\theta}^{2},
  \end{equation}
  which is non-degenerate; hence the simplified roller racer is a weakly degenerate Chaplygin system.

  Therefore, the dynamics of the variables $\theta$ and $\phi$ are specified by the equations of motion, which together with the (nonholonomic) constraints, Eq.~\eqref{eq:Constraints-RollerRacer}, determine the time evolution of the variables $x$ and $y$.
\end{example}

\begin{example}[Bicycle; see \citet{Ge1994}, \citet{GeMa1995}, and \citet{KoMa1997c}]
  \label{ex:Bicycle}
  Consider the simplified model of a bicycle shown in Fig.~\ref{fig:Bicycle}:
  For the sake of simplicity, the wheels are assumed to be massless, and the mass $m$ of the bicycle is considered to be concentrated at a single point; however we take into account the moment of inertia of the steering wheel.

  The configuration space is $Q = SE(2) \times \mathbb{S}^{1} \times \mathbb{S}^{1} = \{ (x, y, \theta, \phi, \psi) \}$; the variables $x$, $y$, $\theta$, and $\psi$ are defined as in Fig.~\ref{fig:Bicycle} and $\phi \defeq \tan \sigma/b$; also let $J(\phi,\psi)$ be the moment of inertia associated with the steering action.
  The Lagrangian $L: TQ \to \R$ is given by
  \begin{multline*}
    L = \frac{m}{2}\left[
      (\cos\theta\,\dot{x} + \sin\theta\,\dot{y} + a \sin\psi\,\dot{\theta})^{2} +
      (\sin\theta\,\dot{x} - \cos\theta\,\dot{y} + a\cos\psi\,\dot{\psi} -c\,\dot{\theta})^{2}
    \right.
    \\
    \left. + a^{2}\sin\psi\,\dot{\psi}^{2}
    \right]
    + \frac{J(\phi,\psi)}{2}\,\dot{\phi}^{2} - m g a \cos\psi,
  \end{multline*}
  which is degenerate.
  \begin{figure}[htbp]
    \centering
    \includegraphics[width=.6\linewidth]{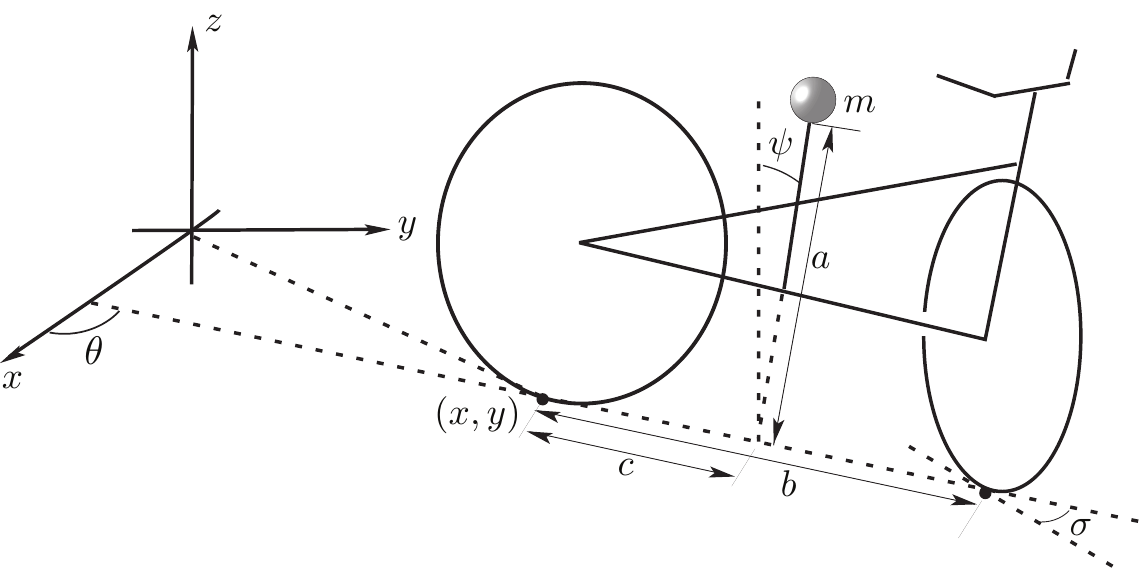}
    \caption{Bicycle (taken from \citet{KoMa1997c} with permission from Wang Sang Koon).}
    \label{fig:Bicycle}
  \end{figure}
  The constraints are given by
  \begin{equation*}
    \dot{\theta} = \phi ( \cos\theta\,\dot{x} + \sin\theta\,\dot{y} ),
    \qquad
    \sin\theta\,\dot{x} - \cos\theta\,\dot{y} = 0.
  \end{equation*}
  Defining the constraint one-forms
  \begin{equation*}
    \omega^{1} \defeq \phi ( \cos\theta\,dx + \sin\theta\,dy ),
    \qquad
    \omega^{2} \defeq \sin\theta\,dx - \cos\theta\,dy,
  \end{equation*}
  we can write the constraint distribution $\Delta_{Q} \subset TQ$ as
  \begin{equation*}
    \Delta_{Q} = \setdef{ \dot{q} = (\dot{x}, \dot{y}, \dot{\theta}, \dot{\phi}, \dot{\psi}) \in TQ }{ \omega^{a}(\dot{q}) = 0,\, a = 1,2 }.
  \end{equation*}

  Let $G = \R^{2}$ and consider the action of $G$ on $Q$ by translations on the $x$-$y$ plane, i.e.,
  \begin{equation*}
    G \times Q \to Q;
    \quad
    \parentheses{(a, b), (x, y, \theta, \phi, \psi)} \mapsto (x + a, y + b, \theta, \phi, \psi).
  \end{equation*}
  Then, the tangent space to the group orbit is given by
  \begin{equation*}
    T_{q}\mathcal{O}(q) = \Span\braces{ \pd{}{x}, \pd{}{y} },
  \end{equation*}
  with $q = (x, y, \theta, \phi, \psi)$.
  It is easy to check that this defines a Chaplygin system in the sense of Definition~\ref{def:ChaplyginSystems}.
  The quotient space is $\bar{Q} \defeq Q/G = \{(\theta, \phi, \psi)\}$, and the horizontal lift $\hl^{\Delta}$ is
  \begin{equation*}
    \hl^{\Delta}_{q}(\dot{\theta}, \dot{\phi}, \dot{\psi}) = \parentheses{
      \frac{\dot{\theta}}{\phi}\,\cos\theta,
      \frac{\dot{\theta}}{\phi}\,\sin\theta,
      \dot{\theta},
      \dot{\phi},
      \dot{\psi}
    }.
  \end{equation*}
  Hence, the reduced Lagrangian $\bar{L}: T\bar{Q} \to \R$ is given by
  \begin{equation*}
    \bar{L} = \frac{m}{2}\brackets{
      (c\,\dot{\theta} - a \cos\psi\,\dot{\psi})^{2}
      + \frac{(\dot{\theta} + a \sin\psi\,\dot{\theta})^{2}}{\phi^{2}}
      + a^{2}\sin\psi\,\dot{\psi}^{2}
    }
    + \frac{J(\phi,\psi)}{2}\,\dot{\phi}^{2} - m g a \cos\psi,
  \end{equation*}
  which is non-degenerate, and so this is a weakly degenerate Chaplygin system as well.
\end{example}

\section{Hamilton--Jacobi Theory for Lagrange--Dirac systems}
\label{sec:DiracHJ}
\subsection{Hamilton--Jacobi Theorem for Lagrange--Dirac systems}
We now state the main theorem of this paper, which relates the
dynamics of the Lagrange--Dirac system with what we refer to as
the {\em Dirac--Hamilton--Jacobi equation}.
\begin{theorem}[Dirac--Hamilton--Jacobi Theorem]
  \label{thm:DiracHJ}
  Suppose that a Lagrangian $L: TQ \to \R$ and a distribution $\Delta_{Q} \subset TQ$ are given.
  Define $\Upsilon: Q \to TQ\oplus T^{*}Q$ by
  \begin{equation*}
    \Upsilon(q) \defeq \mathcal{X}(q) \oplus \gamma(q),
  \end{equation*}
  with a vector field $\mathcal{X}: Q \to TQ$ and a one-form $\gamma: Q \to T^{*}Q$, and assume that it satisfies
  \begin{equation}
    \label{assmptn:Upsilon}
    \Upsilon(q) \in \mathcal{K}_{q} \text{ for any } q \in Q,
  \end{equation}
  and
  \begin{equation}
    \label{eq:dgamma}
    d\gamma|_{\Delta_{Q}} = 0, \ \text{i.e.,}\ d\gamma(v,w) = 0 \text{ for any } v, w \in \Delta_{Q}.
  \end{equation}
  Then, the following are equivalent:
  \begin{enumerate}
    \renewcommand{\theenumi}{\roman{enumi}}
    \renewcommand{\labelenumi}{(\theenumi)}
  \item \label{enumi:DiracHJ-1-i}
    For every integral curve $c(t)$ of $\mathcal{X}$, i.e., for every curve $c: \R \to Q$ satisfying
    \begin{equation}
      \label{eq:DiracHJ-curve}
      \dot{c}(t) = \mathcal{X}( c(t) ),
    \end{equation}
    the curve $t \mapsto \Upsilon \comp c(t) = (\mathcal{X} \oplus \gamma)\comp c(t)$ is an integral curve of the Lagrange--Dirac equations~\eqref{eq:LDEq}.
    \smallskip
  \item \label{enumi:DiracHJ-1-ii}
    $\Upsilon$ satisfies the following {\em Dirac--Hamilton--Jacobi equation}:
    \begin{equation}
      \label{eq:DiracHJ}
      d(\mathcal{E} \comp \Upsilon) \in \Delta_{Q}^{\circ},
    \end{equation}
    or, if $Q$ is connected and $\Delta_{Q}$ is completely nonholonomic\footnote{A distribution $\Delta_{Q} \subset TQ$ is said to be {\em completely nonholonomic} (or {\em bracket-generating}) if $\Delta_{Q}$ along with all of its iterated Lie brackets $[\Delta_{Q}, \Delta_{Q}], [\Delta_{Q}, [\Delta_{Q}, \Delta_{Q}]], \dots$ spans the tangent bundle $TQ$. See, e.g., \citet{VeGe1988} and \citet{Mo2002}.},
    \begin{equation}
      \label{eq:DiracHJ-2}
      \mathcal{E} \comp \Upsilon = E,
    \end{equation}
    with a constant $E$.
  \end{enumerate}
\end{theorem}

\begin{proof}
  Let us first show that \eqref{enumi:DiracHJ-1-ii} implies \eqref{enumi:DiracHJ-1-i}.
  Assume \eqref{enumi:DiracHJ-1-ii} and let $c(t)$ be an integral curve of $\mathcal{X}$, and then set
  \begin{equation*}
    v(t) \oplus p(t) \defeq \Upsilon \comp c(t) = (\mathcal{X} \oplus \gamma) \comp c(t).
  \end{equation*}
  Then, clearly $v(t) = \dot{c}(t) = \mathcal{X}(c(t))$.
  Also, Eq.~\eqref{assmptn:Upsilon} implies that
  \begin{equation*}
    v(t) \in \Delta_{Q}(c(t)),
    \quad
    p(t) = \pd{L}{v}(q(t),v(t)).
  \end{equation*}
  So it remains to show $\dot{p} - \tpd{L}{q} \in \Delta_{Q}^{\circ}$.
  To that end, first calculate
  \begin{equation*}
    \dot{p}_{j}(t) = \od{}{t}\gamma_{j} \comp c(t)
    = \pd{\gamma_{j}}{q^{i}}(c(t))\,\dot{c}^{i}(t)
    = \pd{\gamma_{j}}{q^{i}}(c(t))\,\mathcal{X}^{i}(c(t))
  \end{equation*}
  and so, for any $w \in \Delta_{Q}$, we have
  \begin{equation}
    \dot{p}_{j}(t) w^{j} = \pd{\gamma_{j}}{q^{i}}(c(t))\,\mathcal{X}^{i}(c(t))\,w^{j} = \pd{\gamma_{i}}{q^{j}}(c(t))\,\mathcal{X}^{i}(c(t))\,w^{j},
    \label{eq:dotp,w}
  \end{equation}
  since Eq.~\eqref{eq:dgamma} implies, for any $v, w \in \Delta_{Q}$,
  \begin{equation*}
    \pd{\gamma_{i}}{q^{j}}v^{i}w^{j} = \pd{\gamma_{j}}{q^{i}}v^{i}w^{j},
  \end{equation*}
  and also Eq.~\eqref{assmptn:Upsilon} gives $\mathcal{X}(q) \in \Delta_{Q}(q)$.
  On the other hand,
  \begin{align*}
    d( \mathcal{E} \comp \Upsilon) &= d( \gamma_{i}(q)\, \mathcal{X}^{i}(q) - L(q,\mathcal{X}(q)) )
    \\
    &= \parentheses{
      \pd{\gamma_{i}}{q^{j}}\mathcal{X}^{i} + \gamma_{i} \pd{\mathcal{X}^{i}}{q^{j}}
      - \pd{L}{q^{j}} - \pd{L}{v^{i}}\pd{\mathcal{X}^{i}}{q^{j}}
    }dq^{j}
    \\
    &= \parentheses{
      \pd{\gamma_{i}}{q^{j}}\mathcal{X}^{i} - \pd{L}{q^{j}}
      }dq^{j},
  \end{align*}
  where we used the following relation that follows from Eq.~\eqref{assmptn:Upsilon}:
  \begin{equation*}
    \gamma(q) = \pd{L}{v}(q, \mathcal{X}(q)).
  \end{equation*}
  So the Dirac--Hamilton--Jacobi equation~\eqref{eq:DiracHJ} with Eq.~\eqref{eq:dotp,w} implies
  \begin{align*}
    d( \mathcal{E} \comp \Upsilon)(c(t)) \cdot w &= \parentheses{
      \pd{\gamma_{i}}{q^{j}}(c(t))\, \mathcal{X}^{i}(c(t)) - \pd{L}{q^{j}}(c(t),v(t))
      }w^{j}
      \\
      &= \parentheses{
      \dot{p}_{j}(t) - \pd{L}{q^{j}}(c(t),v(t))
      }w^{j}
      = 0.
  \end{align*}
  Since $w \in \Delta_{Q}$ is arbitrary, this implies
  \begin{equation*}
    \dot{p}(t) - \pd{L}{q}(c(t),v(t)) \in \Delta_{Q}^{\circ}.
  \end{equation*}
  Therefore, \eqref{enumi:DiracHJ-1-i} is satisfied.

  Conversely, assume \eqref{enumi:DiracHJ-1-i}; let $c(t)$ be a curve in $Q$ that satisfies Eq.~\eqref{eq:DiracHJ-curve} and set $v(t) \oplus p(t) \defeq \Upsilon \comp c(t) = (\mathcal{X} \oplus \gamma) \comp c(t)$.
  Then, by assumption, $(c(t), v(t), p(t))$ is an integral curve of the Lagrange--Dirac system~\eqref{eq:LDS}, and so
  \begin{equation*}
    \dot{p}(t) - \pd{L}{q}(c(t),v(t)) \in \Delta_{Q}^{\circ}(c(t)).
  \end{equation*}
  Following the same calculations as above we have, for any $w \in \Delta_{Q}$,
  \begin{equation*}
    d( \mathcal{E} \comp \Upsilon )(c(t)) \cdot w
    = \parentheses{
      \dot{p}_{j}(t) - \pd{L}{q^{j}}(c(t),v(t))
    }w^{j}
    = 0.
  \end{equation*}
  For an arbitrary point $q \in Q$, we can consider an integral curve $c(t)$ of $X$ such that $c(0) = q$.
  Therefore, the above equation implies that $d( \mathcal{E} \comp \Upsilon )(q) \cdot w_{q} = 0$ for any $q \in Q$ and $w_{q} \in \Delta_{Q}(q)$, which gives the Dirac--Hamilton--Jacobi equation~\eqref{eq:DiracHJ}.
  If $Q$ is connected and $\Delta_{Q}$ is completely nonholonomic, then by the same argument as in the proof of Theorem~3.1 in \citet{OhBl2009}, $d( \mathcal{E} \comp \Upsilon ) \in \Delta_{Q}^{\circ}$ reduces to $\mathcal{E} \comp \Upsilon = E$ for some constant $E$.
\end{proof}

Theorem~\ref{thm:DiracHJ} can be recast in the context of
Section~\ref{ssec:LDSonPontryaginBundle} as follows:
\begin{corollary}
  Under the same conditions as in Theorem~\ref{thm:DiracHJ}, the following are equivalent:
  \begin{enumerate}
    \renewcommand{\theenumi}{\roman{enumi}}
    \renewcommand{\labelenumi}{(\theenumi)}
  \item \label{enumi:DiracHJ-2-i}
    For every curve $c(t)$ such that
    \begin{equation*}
      \dot{c}(t)=T\pr_Q \cdot \tilde{X}(\Upsilon \comp c(t)),
    \end{equation*}
    the curve $t \mapsto \Upsilon \comp c(t)$ is an integral curve of $\tilde{X}$, and so it is an integral curve of the Lagrange--Dirac equations~\eqref{eq:LDEq}.
    \medskip
  \item \label{enumi:DiracHJ-2-ii}
    $\Upsilon$ satisfies $(0, d(\mathcal{E}\comp\Upsilon\comp \pr_Q)) \in D_{TQ \oplus T^{*}Q}$, or equivalently, $d(\mathcal{E}\comp\Upsilon\comp \pr_Q) \in \Delta_{TQ\oplus T^*Q}^{\circ}$.
  \end{enumerate}
\end{corollary}

\begin{proof}
  The equivalence of \eqref{enumi:DiracHJ-2-i} with that of Theorem~\ref{thm:DiracHJ} follows from the relation $T\pr_{Q} \comp \tilde{X} \comp \Upsilon = \mathcal{X}$, which is easily checked by coordinate calculations.

  On the other hand, for \eqref{enumi:DiracHJ-2-ii}, first observe that $\pr_{Q}^{*}(\Delta_{Q}^{\circ}) = \Delta_{TQ \oplus T^{*}Q}^{\circ}$.
  Then, since $\pr_{Q}: TQ \oplus T^{*}Q \to Q$ is a surjective submersion, it follows that
  \begin{equation*}
    d(\mathcal{E} \comp \Upsilon) \in \Delta_{Q}^{\circ}
    \iff
    \pr_{Q}^{*} d(\mathcal{E} \comp \Upsilon) \in \pr_{Q}^{*}(\Delta_{Q}^{\circ})
    \iff
    d(\mathcal{E} \comp \Upsilon \comp \pr_{Q}) \in \Delta_{TQ \oplus T^{*}Q}^{\circ}.
  \end{equation*}
  This proves the equivalence of \eqref{enumi:DiracHJ-2-ii} with that of Theorem~\ref{thm:DiracHJ}.
\end{proof}

\subsection{Nonholonomic Hamilton--Jacobi Theory as a Special Case}
Let us show that the nonholonomic Hamilton--Jacobi equation of \citet{IgLeMa2008} and \citet{OhBl2009} follows as a special case of the above theorem.
Consider the special case where the Lagrangian $L: TQ \to \R$ is non-degenerate, i.e., the Legendre transformation $\Leg: TQ \to T^{*}Q$ is invertible.
Then, we may rewrite the definition of the submanifold $\mathcal{K} \subset TQ \oplus T^{*}Q$, Eq.~\eqref{eq:mathcalK}, by
\begin{align*}
  \mathcal{K} &= \setdef{ v_{q} \oplus p_{q} \in TQ \oplus T^{*}Q }{ v_{q} \in \Delta_{Q}(q),\ p_{q} = \Leg(v_{q}) }
  \\
  &= \setdef{ v_{q} \oplus p_{q} \in TQ \oplus T^{*}Q }{ p_{q} \in P_{q},\ v_{q} = (\Leg)^{-1}(p_{q}) }
  \\
  &= \Delta_{Q} \oplus P,
\end{align*}
where we recall that $P \defeq \Leg(\Delta_{Q})$.
It implies that if $\Upsilon = \mathcal{X} \oplus \gamma$ takes values in $\mathcal{K}$ then $\mathcal{X} = (\Leg)^{-1} \comp \gamma$, and thus
\begin{equation*}
  \mathcal{E} \comp \Upsilon(q)
  = \ip{ \gamma(q) }{ (\Leg)^{-1}(\gamma(q)) } - L \comp (\Leg)^{-1}(\gamma(q))
  = H \comp \gamma(q),
\end{equation*}
with $\gamma$ taking values in $P$ and the Hamiltonian $H: T^{*}Q \to \R$ defined by
\begin{equation*}
  H(q, p) \defeq \ip{ p_{q} }{ (\Leg)^{-1}(p_{q}) } - L \comp (\Leg)^{-1}(p_{q}).
\end{equation*}
Then, the Lagrange--Dirac equations~\eqref{eq:LDEq} become the nonholonomic Hamilton's equations
\begin{equation*}
  \dot{q} = \frac{\partial H}{\partial p}(q,p),
  \qquad
  \dot{p} + \frac{\partial H}{\partial q}(q,p) \in \Delta_Q^\circ(q).
  \qquad
  \dot{q} \in \Delta_Q(q),
\end{equation*}
or, in an intrinsic form,
\begin{equation*}
  i_{X_{H}^{\rm nh}} \Omega - dH \in \Delta_{T^{*}Q}^{\circ},
  \qquad
  T\pi_{Q}(X_{H}^{\rm nh}) \in \Delta_{Q}
\end{equation*}
for a vector field $X_{H}^{\rm nh}$ on $T^{*}Q$.
Furthermore, it is straightforward to show that
\begin{equation*}
  (\Leg)^{-1} = \mathbb{F}H = T\pi_{Q} \comp X_{H},
\end{equation*}
where $X_{H}$ is the Hamiltonian vector field of the unconstrained system with the same Hamiltonian, i.e., $i_{X_{H}}\Omega = dH$; hence we obtain
\begin{equation*}
  \mathcal{X}(q) = (\Leg)^{-1} \comp \gamma(q) =  T\pi_{Q} \cdot X_{H}(\gamma(q)).
\end{equation*}
Therefore, Theorem~\ref{thm:DiracHJ} specializes to the nonholonomic Hamilton--Jacobi theorem of \citet{IgLeMa2008} and \citet{OhBl2009}:

\begin{corollary}[Nonholonomic Hamilton--Jacobi~\cite{IgLeMa2008, OhBl2009}]
  Consider a nonholonomic system defined on a configuration manifold $Q$ with a Lagrangian of the form Eq.~\eqref{eq:SimpleLagrangian} and a nonholonomic constraint distribution $\Delta_{Q} \subset TQ$.
  Let $\gamma: Q \to T^{*}Q$ be a one-form that satisfies
  \begin{equation*}
    \gamma(q) \in P_{q} \text{ for any } q \in Q,
  \end{equation*}
  and
  \begin{equation*}
    d\gamma|_{\Delta_{Q}} = 0, \text{ i.e., } d\gamma(v,w) = 0 \text{ for any } v, w \in \Delta_{Q}.
  \end{equation*}
  Then, the following are equivalent:
  \begin{enumerate}
    \renewcommand{\theenumi}{\roman{enumi}}
    \renewcommand{\labelenumi}{(\theenumi)}
  \item
    For every curve $c(t)$ in $Q$ satisfying
    \begin{equation*}
      \dot{c}(t) = T\pi_{Q} \cdot X_{H}( \gamma \comp c(t) ),
    \end{equation*}
    the curve $t \mapsto \gamma \comp c(t)$ is an integral curve of $X_{H}^{\rm nh}$, where $X_{H}$ is the Hamiltonian vector field of the unconstrained system with the same Hamiltonian, i.e., $i_{X_{H}}\Omega = dH$.
    \medskip
  \item
    The one-form $\gamma$ satisfies the {\em nonholonomic Hamilton--Jacobi equation}:
    \begin{equation*}
      d(H \comp \gamma) \in \Delta_{Q}^{\circ},
    \end{equation*}
    or, if $Q$ is connected and $\Delta_{Q}$ is completely nonholonomic,
    \begin{equation*}
      H \comp \gamma = E,
    \end{equation*}
    with a constant $E$.
  \end{enumerate}
\end{corollary}

\subsection{Applications to Degenerate Lagrangian System with Holonomic Constraints}
If the constraints are holonomic, then the distribution
$\Delta_{Q} \subset TQ$ is integrable, and so there exists a local
submanifold $S \subset Q$ such that $T_{s}S = \Delta_{Q}(s)$ for
any $s \in S$. Let $\iota_{S}: S \hookrightarrow Q$ be the
inclusion. Then, the Dirac--Hamilton--Jacobi
equation~\eqref{eq:DiracHJ} gives
\begin{equation*}
  \iota_{S}^{*} d(\mathcal{E} \comp \Upsilon) \in (TS)^{\circ} = 0,
\end{equation*}
and thus
\begin{equation*}
  d(\mathcal{E} \comp \Upsilon \comp \iota_{S}) = 0,
\end{equation*}
which implies that we have
\begin{equation}
  \label{eq:DiracHJ-Holonomic}
  \mathcal{E} \comp \Upsilon \comp \iota_{S} = E,
\end{equation}
with a constant $E$, assuming $S$ is connected.

On the other hand, the condition~\eqref{eq:dgamma} becomes
\begin{equation}
  \label{eq:dgamma-Holonomic}
  \iota_{S}^{*} d\gamma = d(\gamma \comp \iota_{S}) = 0,
\end{equation}
and so $\gamma \comp \iota_{S} = dW$ for some function $W$ defined
locally on $S$.

\begin{example}[LC circuit; see \citet{YoMa2006a,YoMa2007b}]
  Consider the LC circuit shown in Fig.~\ref{fig:LCC}.
  \begin{figure}[htbp]
    \centering
    \includegraphics[width=.4\linewidth]{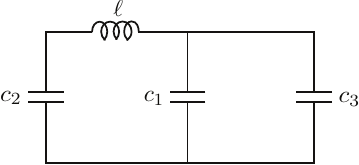}
    \caption{LC circuit (see \citet{YoMa2006a}).}
    \label{fig:LCC}
  \end{figure}
  The configuration space is the 4-dimensional vector space $Q = \{(q_{\ell}, q_{c_{1}}, q_{c_{2}}, q_{c_{3}})\}$, which represents charges in the circuit elements. Then $TQ \cong Q \times Q$ and $f_{q} = (f_{\ell}, f_{c_{1}}, f_{c_{2}}, f_{c_{3}}) \in T_{q}Q$ represents the currents in the corresponding circuit elements.
  The Lagrangian $L: TQ \to \R$ is given by
  \begin{equation*}
    L(q,f) = \frac{1}{2} \ell\,f_{\ell}^{2} - \frac{1}{2}\frac{q_{c_{1}}^{2}}{c_{1}} - \frac{1}{2}\frac{q_{c_{2}}^{2}}{c_{2}} - \frac{1}{2}\frac{q_{c_{3}}^{2}}{c_{3}},
  \end{equation*}
  which is clearly degenerate.

  The generalized energy $\mathcal{E}: TQ \oplus T^{*}Q \to \R$ is
  \begin{align*}
    \mathcal{E}(q,f,p) &=  p \cdot f - L(q,f)
    \\
    &= p_{\ell}f_{\ell} + p_{c_{1}}f_{c_{1}} + p_{c_{2}}f_{c_{2}} + p_{c_{3}}f_{c_{3}}
    - \frac{1}{2} \ell\,f_{\ell}^{2} + \frac{1}{2}\frac{q_{c_{1}}^{2}}{c_{1}} + \frac{1}{2}\frac{q_{c_{2}}^{2}}{c_{2}} + \frac{1}{2}\frac{q_{c_{3}}^{2}}{c_{3}}.
  \end{align*}

  The Kirchhoff current law gives the constraints $-f_{\ell} + f_{c_{2}} = 0$ and $f_{c_{1}} - f_{c_{2}} + f_{c_{3}} = 0$, or in terms of constraint one-forms, $\omega^{1} = -dq_{\ell} + dq_{c_{2}}$ and $\omega^{2} = dq_{c_{1}} - dq_{c_{2}} + dq_{c_{3}}$.
  Thus, the constraint distribution $\Delta_{Q} \subset TQ$ is given by
  \begin{equation*}
    \Delta_{Q} = \setdef{ f \in TQ }{ \omega^{a}(f) = 0,\, a = 1,2 }.
  \end{equation*}
  So the submanifold $\mathcal{K} \subset TQ \oplus T^{*}Q$ is
  \begin{equation*}
    \mathcal{K} =
    \setdef{ (q, f, p) \in TQ \oplus T^{*}Q }{
      f_{\ell} = f_{c_{2}},\   f_{c_{2}} = f_{c_{1}} + f_{c_{3}},\
      p_{\ell} = \ell\,f_{\ell},\  p_{c_{1}} = p_{c_{2}} = p_{c_{3}} = 0
      }.
  \end{equation*}
  Hence, the generalized energy constrained to $\mathcal{K}$ is
  \begin{equation*}
    \mathcal{E}|_{\mathcal{K}} = \frac{1}{2} \ell\,f_{\ell}^{2} + \frac{1}{2}\frac{q_{c_{1}}^{2}}{c_{1}} + \frac{1}{2}\frac{q_{c_{2}}^{2}}{c_{2}} +
    \frac{1}{2}\frac{q_{c_{3}}^{2}}{c_{3}}.
  \end{equation*}

  Notice that the constraints are holonomic, i.e., the constraints can be integrated to give
  \begin{equation*}
    q_{\ell} - q_{c_{2}} = a_{0},
    \quad
    q_{c_{2}} - q_{c_{1}} - q_{c_{3}} = a_{1},
  \end{equation*}
  with some constants $a_{0}$ and $a_{1}$.
  So we define a submanifold $S \subset Q$ by
  \begin{equation*}
    S \defeq \setdef{ (q_{\ell}, q_{c_{1}}, q_{c_{2}}, q_{c_{3}}) \in Q }{  q_{c_{2}} = q_{\ell} - a_{0},\; q_{c_{3}} = q_{c_{2}} - q_{c_{1}} - a_{1} } = \{ (q_{\ell}, q_{c_{1}} ) \},
  \end{equation*}
  and the inclusion
  \begin{equation*}
    \iota_{S}: S \hookrightarrow Q;
    \quad
    (q_{\ell}, q_{c_{1}}) \mapsto ( q_{\ell}, q_{c_{1}}, q_{\ell} - a_{0}, q_{c_{2}} - q_{c_{1}} - a_{1} ).
  \end{equation*}

  Now, the Dirac--Hamilton--Jacobi equation for holonomic systems, Eq.~\eqref{eq:DiracHJ-Holonomic}, gives
  \begin{equation}
    \label{eq:DiracHJ-LCC}
    \mathcal{E} \comp \Upsilon \comp \iota_{S} = E,
  \end{equation}
  with some constant $E$, where $\Upsilon \comp \iota_{S}: S \to TQ \oplus T^{*}Q$ is
  \begin{equation*}
    \Upsilon \comp \iota_{S}(q_{\ell}, q_{c_{1}})
    = \parentheses{
      q_{\ell}, q_{c_{1}}, \tilde{\mathcal{X}}(q_{\ell}, q_{c_{1}}), \tilde{\gamma}(q_{\ell}, q_{c_{1}})
      }
  \end{equation*}
  with $\tilde{\mathcal{X}} \defeq \mathcal{X} \comp \iota_{S}: S \to TQ$ and $\tilde{\gamma} \defeq \gamma \comp \iota_{S}: S \to T^{*}Q$ given by
  \begin{align*}
    \tilde{\mathcal{X}}(q_{\ell}, q_{c_{1}})
    &= \parentheses{
      \tilde{\mathcal{X}}_{\ell}(q_{\ell}, q_{c_{1}}), \tilde{\mathcal{X}}_{c_{1}}(q_{\ell}, q_{c_{1}}), \tilde{\mathcal{X}}_{c_{2}}(q_{\ell}, q_{c_{1}}), \tilde{\mathcal{X}}_{c_{3}}(q_{\ell}, q_{c_{1}})
    },
    \\
    \tilde{\gamma}(q_{\ell}, q_{c_{1}})
    &= \parentheses{
      \tilde{\gamma}_{\ell}(q_{\ell}, q_{c_{1}}), \tilde{\gamma}_{c_{1}}(q_{\ell}, q_{c_{1}}), \tilde{\gamma}_{c_{2}}(q_{\ell}, q_{c_{1}}), \tilde{\gamma}_{c_{3}}(q_{\ell}, q_{c_{1}})
    }.
  \end{align*}
  The condition $\Upsilon \comp \iota_{S}(q_{\ell}, q_{c_{1}}) \in \mathcal{K}$ implies
  \begin{equation*}
    \tilde{\mathcal{X}}_{\ell} = \tilde{\mathcal{X}}_{c_{2}},
    \quad
    \tilde{\mathcal{X}}_{c_{2}} = \tilde{\mathcal{X}}_{c_{1}} + \tilde{\mathcal{X}}_{c_{3}},
    \quad
    \tilde{\gamma}_{\ell} = \ell\,\tilde{\mathcal{X}}_{\ell},
    \quad
    \tilde{\gamma}_{c_{1}} = \tilde{\gamma}_{c_{2}} = \tilde{\gamma}_{c_{3}} = 0.
  \end{equation*}
  Then,
  \begin{equation*}
    \tilde{\gamma} = \gamma \comp \iota_{S} = \ell\,\tilde{\mathcal{X}}_{\ell}(q_{\ell}, q_{c_{1}})\,dq_{\ell},
  \end{equation*}
  and thus condition~\eqref{eq:dgamma-Holonomic} gives
  \begin{equation*}
    \pd{\tilde{\mathcal{X}}_{\ell}}{q_{c_{1}}} = 0,
  \end{equation*}
  and hence $\tilde{\mathcal{X}}_{\ell}(q_{\ell}, q_{c_{1}}) = \tilde{\mathcal{X}}_{\ell}(q_{\ell})$.
  The Dirac--Hamilton--Jacobi equation~\eqref{eq:DiracHJ-LCC} then becomes
  \begin{equation}
    \label{eq:DiracHJ-LCC-1}
    \frac{1}{2} \ell\, \tilde{\mathcal{X}}_{\ell}(q_{\ell})^{2}
    + \frac{1}{2}\frac{q_{c_{1}}^{2}}{c_{1}}
    + \frac{1}{2}\frac{(q_{\ell} - a_{0})^{2}}{c_{2}}
    + \frac{1}{2}\frac{(q_{\ell} - q_{c_{1}} - a_{0} - a_{1})^{2}}{c_{3}}
    = E.
  \end{equation}
  We impose the condition that $\mathcal{X}_{\ell} = 0$ when $q_{\ell} = q_{c_{1}} = 0$ and $E = 0$, which corresponds to the case where nothing is happening in the circuit.
  Then, we have
  \begin{equation*}
    \frac{a_{0}^{2}}{c_{2}}
    + \frac{(a_{0} + a_{1})^{2}}{c_{3}}
    = 0,
  \end{equation*}
  which gives $a_{0} = a_{1} = 0$, since $c_{2}$ and $c_{3}$ are both positive.
  Therefore, Eq.~\eqref{eq:DiracHJ-LCC-1} becomes
  \begin{equation}
    \label{eq:DiracHJ-LCC-2}
    \frac{1}{2} \ell\, \tilde{\mathcal{X}}_{\ell}(q_{\ell})^{2}
    + \frac{1}{2}\frac{q_{c_{1}}^{2}}{c_{1}}
    + \frac{1}{2}\frac{q_{\ell}^{2}}{c_{2}}
    + \frac{1}{2}\frac{(q_{\ell} - q_{c_{1}})^{2}}{c_{3}}
    = E.
  \end{equation}
  Taking the derivative with respect to $q_{c_{1}}$ of both sides and solving for $q_{c_{1}}$, we have
  \begin{equation*}
    q_{c_{1}} = \frac{c_{1}}{c_{1} + c_{3}}\, q_{\ell}.
  \end{equation*}
  Substituting this into Eq.~\eqref{eq:DiracHJ-LCC-2} gives
  \begin{equation*}
    \frac{1}{2} \parentheses{
      \ell\,\tilde{\mathcal{X}}_{\ell}(q_{\ell})^{2}
      + \frac{c_{1} + c_{2} + c_{3}}{c_{2}(c_{1} + c_{3})}\,q_{\ell}^{2}
    }
    = E.
  \end{equation*}
  Solving for $\tilde{\mathcal{X}}_{\ell}(q_{\ell})$, we obtain
  \begin{equation*}
    \mathcal{X}_{\ell}(q) = \tilde{\mathcal{X}}_{\ell}(q_{\ell}) = \pm \sqrt{ \frac{1}{\ell} \parentheses{ 2E - \frac{c_{1} + c_{2} + c_{3}}{c_{2}(c_{1} + c_{3})}\,q_{\ell}^{2} } }.
  \end{equation*}
  Taking the positive root, Eq.~\eqref{eq:DiracHJ-curve} for $q_{\ell}$ gives
  \begin{equation*}
    \dot{q}_{\ell} = \sqrt{ \frac{1}{\ell} \parentheses{ 2E - \frac{c_{1} + c_{2} + c_{3}}{c_{2}(c_{1} + c_{3})}\,q_{\ell}^{2} } },
  \end{equation*}
  which can be solved easily:
  \begin{equation*}
    q_{\ell}(t) = \sqrt{\frac{2E}{\ell\,\nu^{2}}}\,\sin(\nu t + \alpha),
  \end{equation*}
  where
  \begin{equation*}
    \nu \defeq \sqrt{ \frac{c_{1} + c_{2} + c_{3}}{c_{2}(c_{1} + c_{3})\,\ell} }
  \end{equation*}
  and $\alpha$ is a phase constant to be determined by the initial condition.
\end{example}

\begin{remark}
  In the conventional LC circuit theory, one often simplifies problems by ``combining'' capacitors.
  Using this technique, the above example simplifies to an LC circuit with an inductor with inductance $\ell$ and a single capacitance $C$, that satisfies the following equation:
  \begin{equation*}
    \frac{1}{C} = \frac{1}{c_{2}} + \frac{1}{c_{1} + c_{3}},
  \end{equation*}
  which gives
  \begin{equation*}
    C = \frac{c_{2}(c_{1} + c_{3})}{c_{1} + c_{2} + c_{3}}.
  \end{equation*}
  Then, the equation for the current $i_{\ell} \defeq \dot{q}_{\ell}$ is given by
  \begin{equation*}
    \ell\,\od{^{2}{i}_{\ell}}{t^{2}} + \frac{i_{\ell}}{C} = 0,
  \end{equation*}
  or
  \begin{equation*}
    \od{^{2}{i}_{\ell}}{t^{2}} + \nu\, i_{\ell} = 0,
  \end{equation*}
  with
  \begin{equation*}
    \nu = \frac{1}{\sqrt{\ell\,C}} = \sqrt{ \frac{c_{1} + c_{2} + c_{3}}{c_{2}(c_{1} + c_{3})\,\ell} },
  \end{equation*}
  which coincides the one defined above.
  The general solution of the above ODE is
  \begin{equation*}
    i_{\ell}(t) = \dot{q}_{\ell}(t) = A \sin(\nu t + \alpha)
  \end{equation*}
  for some constants $A$ and $\alpha$.
  Therefore, our solution is consistent with the conventional theory.
\end{remark}

\subsection{Applications to Degenerate Lagrangian System with Nonholonomic Constraints}
\begin{example}[Simplified Roller Racer; see Example~\ref{ex:RollerRacer}]
  \label{ex:RollerRacer2}
  The submanifold $\mathcal{K} \subset TQ \oplus T^{*}Q$ is given by
  \begin{multline*}
    \mathcal{K} = \Bigl\{ (q, v, p) \in TQ \oplus T^{*}Q
    \ |\ 
    v_{x} = \cos\theta\csc\phi[ (d_{1}\cos\phi + d_{2})\,v_{\theta} + d_{2} v_{\phi} ],
    \\
    v_{y} = \sin\theta\csc\phi[ (d_{1}\cos\phi + d_{2})\,v_{\theta} + d_{2} v_{\phi} ],\;
    p_{x} = m_{1} v_{x},\;
    p_{y} = m_{1} v_{y},\;
    p_{\theta} = I_{1} v_{\theta},\;
    p_{\phi} = 0
    \Bigr\},
  \end{multline*}
  and the generalized energy constrained to $\mathcal{K}$ is
  \begin{equation*}
    \mathcal{E}|_{\mathcal{K}}
    = \frac{1}{2} m_{1}\csc^{2}\phi \brackets{ (d_{1}\cos\phi + d_{2})\,v_{\theta} + d_{2} v_{\phi} }^{2}
    +  \frac{1}{2} I_{1} v_{\theta}^{2}.
  \end{equation*}

  The distribution $\Delta_{Q}$ is easily shown to be completely nonholonomic, and thus we may use the Dirac--Hamilton--Jacobi equation \eqref{eq:DiracHJ-2}, which gives
  \begin{equation}
    \label{eq:DiracHJ-RollerRacer}
    \frac{1}{2} m_{1}\csc^{2}\phi \brackets{ (d_{1}\cos\phi + d_{2})\,\mathcal{X}_{\theta}(q) + d_{2} \mathcal{X}_{\phi}(q) }^{2}
    +  \frac{1}{2} I_{1} \mathcal{X}_{\theta}(q)^{2} = E.
  \end{equation}

  Now, we assume the following ansatz\footnote{The $(x,y)$-dependence is eliminated because we expect that the vector field $\mathcal{X}$ to be $\R^{2}$-translational invariant since the system possesses $\R^{2}$-symmetry.}:
  \begin{equation}
    \label{eq:ansatz1-RollerRacer}
    \mathcal{X}_{\theta}(x, y, \theta, \phi) = \mathcal{X}_{\theta}(\theta, \phi),
    \qquad
    \mathcal{X}_{\phi}(x, y, \theta, \phi) = \mathcal{X}_{\phi}(\phi).
  \end{equation}
  However, substituting them into Eq.~\eqref{eq:DiracHJ-RollerRacer} and solving for $\mathcal{X}_{\theta}$ shows that $\mathcal{X}_{\theta}$ does not depend on $\theta$ either; hence we set $\mathcal{X}_{\theta}(\theta, \phi) = \mathcal{X}_{\theta}(\phi)$.
  Then, solving Eq.~\eqref{eq:DiracHJ-RollerRacer} for $\mathcal{X}_{\phi}$, we have
  \begin{equation}
    \label{eq:mathcalXphi-RollerRacer}
    \mathcal{X}_{\phi}(\phi) = \frac{ -(d_{1}\cos\phi + d_{2}) \mathcal{X}_{\theta}(\phi) \pm \sin\phi \sqrt{2E - I_{1}\mathcal{X}_{\theta}(\phi)^{2}} }{\sqrt{m_{1}}\,d_{2}}.
  \end{equation}
  Substituting the first solution into condition~\eqref{eq:dgamma}, we have
  \begin{equation*}
    \brackets{
      (d_{1}\cos\phi + d_{2})\,\mathcal{X}_{\theta}(\phi) - \sin\phi \sqrt{ 2E - I_{1}\mathcal{X}_{\theta}(\phi)^{2} }
    } \mathcal{X}_{\theta}'(\phi) = 0.
  \end{equation*}
  We choose $\mathcal{X}_{\theta}'(\phi) = 0$ and hence
  \begin{equation*}
    \mathcal{X}_{\theta}(\phi) = v_{\theta},
  \end{equation*}
  where $v_{\theta}$ is the initial angular velocity in the $\theta$-direction.
  This is consistent with the Lagrange--Dirac equations~\eqref{eq:LDEq-RollerRacer}, which give $\ddot{\theta} = 0$.
  Substituting this into the first case of Eq.~\eqref{eq:mathcalXphi-RollerRacer}, we obtain
  \begin{equation*}
    \mathcal{X}_{\phi}(\phi) = -v_{\theta} \parentheses{ 1 + \frac{d_{1}}{d_{2}} \cos\phi } + \frac{v_{r}}{d_{2}} \sin\phi,
  \end{equation*}
  where $v_{r} \defeq \sqrt{ (2E - I_{1} v_{\theta}^{2})/m_{1}}$.
  
  Then, the condition $\mathcal{X}(q) \in \Delta_{Q}(q)$ gives the other components of the vector field $\mathcal{X}$, and hence Eq.~\eqref{eq:DiracHJ-curve} gives
  \begin{equation*}
    \begin{array}{c}
      \dot{x} = v_{r} \cos\theta,
      \qquad
      \dot{y} = v_{r} \sin\theta,
      \bigskip\\
      \DS \dot{\theta} = 0,
      \qquad
      \DS \dot{\phi} = -v_{\theta} \parentheses{ 1 + \frac{d_{1}}{d_{2}} \cos\phi } + \frac{v_{r}}{d_{2}}
      \sin\phi.
    \end{array}
  \end{equation*}
  We can solve the last equation by separation of variables, and the rest is explicitly solvable.
\end{example}

\subsection{Lagrangians that are Linear in Velocity}
\label{ssec:Limitation} There are some physical systems, such as point vortices~(see, e.g., \citet{Ch1978} and \citet{Ne2001}), which are described by Lagrangians that are linear in velocity, i.e.,
\begin{equation}
  \label{eq:linearLagrangian}
  L(q, \dot{q}) = \ip{ \alpha(q) }{ \dot{q} } - h(q),
\end{equation}
where $\alpha$ is a one-form on $Q$.
The Lagrangian is clearly degenerate and Lagrange--Dirac equations~\eqref{eq:LDEq} give the following equations of motion (see \citet{RoMa2002} and \citet{YoMa2007b}):
\begin{equation}
  \label{eq:alpha-h}
  -i_{\mathcal{X}} d\alpha = dh,
\end{equation}
where $\mathcal{X}$ is a vector field on $Q$; hence the Lagrange--Dirac equations~\eqref{eq:LDEq} reduce to the first-order dynamics $\dot{q} = \mathcal{X}(q)$ defined on $Q$.

Now, the assumption in \eqref{assmptn:Upsilon} of Theorem~\ref{thm:DiracHJ} implies $\gamma(q) = \alpha(q)$ and thus
\begin{equation*}
  \mathcal{E} \comp \Upsilon(q) = h(q);
\end{equation*}
so the Dirac--Hamilton--Jacobi equation~\eqref{eq:DiracHJ} gives
\begin{equation*}
  h(q) = E,
\end{equation*}
which simply defines a level set of the energy of the dynamics on $Q$, i.e., the Dirac--Hamilton--Jacobi equation~\eqref{eq:DiracHJ-2} does not give any information on the dynamics on $Q$.
This is because the original dynamics, which is naturally defined on $Q$ with the one-form $\alpha$ and the function $h$, is somewhat artificially lifted to the tangent bundle $TQ$ through the linear Lagrangian~\eqref{eq:linearLagrangian}.
In fact, for point vortices on the plane, one has $Q = \R^{2}$, and the two-form $-d\alpha$ is a symplectic form; hence $Q = \R^{2}$ is a symplectic manifold and Eq.~\eqref{eq:alpha-h} defines a Hamiltonian system on $Q$ with the Hamiltonian $h$.

\section{Hamilton--Jacobi Theory for Weakly Degenerate Chaplygin Systems}
\label{sec:DiracHJ-WDCS}
In this section, we first show that a weakly Chaplygin system introduced in Section~\ref{ssec:WeaklyDegenerateChaplyginSystems} reduces to an almost Hamiltonian system on $T^{*}\bar{Q}$ with a reduced Hamiltonian $\bar{H}: T^{*}\bar{Q} \to \R$, where $\bar{Q} \defeq Q/G$.
Accordingly, we may consider a variant of the nonholonomic Hamilton--Jacobi equation~\cite{IgLeMa2008,OhBl2009} for the reduced system, which we call the {\em reduced Dirac--Hamilton--Jacobi equation}.
We then show an explicit formula that maps solutions of the reduced Dirac--Hamilton--Jacobi equation to those of the original one.
Thus, one may solve the reduced Dirac--Hamilton--Jacobi equation, which is simpler than the original one, and then construct solutions of the original Dirac--Hamilton--Jacobi equation by the formula.
\subsection{The Geometry of Weakly Degenerate Chaplygin Systems}
\label{ssec:GeometryOfWDCS}
For weakly degenerate Chaplygin systems, the geometric structure introduced in Section~\ref{ssec:ChaplyginSystems} is carried over to the Hamiltonian side.
Specifically, we define the horizontal lift
$\hl^{P}_{q}: T_{\bar{q}}^{*}\bar{Q} \to P_{q}$ by~(see
\citet{EhKoMoRi2004})
\begin{equation*}
  \hl^{P}_{q}
  \defeq \Leg_{q} \comp \hl^{\Delta}_{q} \comp (\F\bar{L})^{-1}_{\bar{q}},
\end{equation*}
or by requiring that the diagram below commutes.
\begin{equation*}
  \vcenter{
    \xymatrix@!0@R=0.675in@C=1in{
      \Delta_{Q}(q) \ar[r]^{\Leg_{q}} & P_{q}
      \\
      T_{\bar{q}}\bar{Q} \ar[u]^{\hl^{\Delta}_{q}} & T_{\bar{q}}^{*}\bar{Q} \ar[l]^{\ (\mathbb{F}\bar{L})^{-1}_{\bar{q}}} \ar@{-->}[u]_{\hl^{P}_{q}}
    }
  }
\end{equation*}
It is easy to show that the following equality holds for the
pairing between the two horizontal lifts (see Lemma~A.1 in
\citet{OhFeBlZe2011}): For any $\alpha_{\bar{q}} \in
T^{*}_{\bar{q}}\bar{Q}$ and $v_{\bar{q}} \in T_{\bar{q}}\bar{Q}$,
\begin{equation}
  \label{eq:hl-pairing}
  \ip{ \hl^{P}_{q}(\alpha_{\bar{q}}) }{ \hl^{\Delta}_{q}(v_{\bar{q}}) } = \ip{ \alpha_{\bar{q}} }{ v_{\bar{q}} }.
\end{equation}
We also define a map $\hl^{\mathcal{K}}_{q}:
T_{\bar{q}}^{*}\bar{Q} \to \mathcal{K}_{q} \subset T_{q}Q \oplus
T^{*}_{q}Q$ by
\begin{equation*}
  \hl^{\mathcal{K}}_{q} \defeq \parentheses{ \hl^{\Delta}_{q} \comp (\F\bar{L})^{-1}_{\bar{q}} } \oplus \hl^{P}_{q}.
\end{equation*}

Since the reduced Lagrangian $\bar{L}$ is non-degenerate, we can
also define the reduced Hamiltonian\footnote{Recall that we cannot
define a Hamiltonian $H: T^{*}Q \to \R$ for the original system
because the original Lagrangian $L: TQ \to \R$ is degenerate.}
$\bar{H}: T^{*}\bar{Q} \to \R$ as follows:
\begin{equation}
  \label{eq:barH}
  \bar{H}(p_{\bar{q}}) \defeq \ip{ p_{\bar{q}} }{ v_{\bar{q}} } - \bar{L}(v_{\bar{q}}),
\end{equation}
with $v_{\bar{q}} = (\F\bar{L})^{-1}_{\bar{q}}(p_{\bar{q}})$.

\begin{lemma}
  \label{lem:mathcalE-barH}
  The generalized energy $\mathcal{E}: TQ \oplus T^{*}Q \to \R$ and the reduced Hamiltonian $\bar{H}$ are related as follows:
  \begin{equation*}
    \mathcal{E} \comp \hl^{\mathcal{K}} = \bar{H}.
  \end{equation*}
\end{lemma}

\begin{proof}
  Follows from the following simple calculation: For an arbitrary $\alpha_{\bar{q}} \in T^{*}_{\bar{q}}\bar{Q}$, let $q \in Q$ be a point such that $\pi(q) = \bar{q}$.
  Then, we obtain
  \begin{align*}
    \mathcal{E} \comp \hl^{\mathcal{K}}_{q}( \alpha_{\bar{q}} )
    &= \ip{ \hl^{P}_{q}(\alpha_{\bar{q}}) }{ \hl^{\Delta}_{q} \comp (\F\bar{L})^{-1}_{\bar{q}}(\alpha_{\bar{q}}) }
    - L \comp \hl^{\Delta}_{q} \comp (\F\bar{L})^{-1}_{\bar{q}}(\alpha_{\bar{q}})
    \\
    &= \ip{ \alpha_{\bar{q}} }{ (\F\bar{L})^{-1}_{\bar{q}}(\alpha_{\bar{q}}) }
    - \bar{L} \comp (\F\bar{L})^{-1}_{\bar{q}}(\alpha_{\bar{q}})
    \\
    &= \bar{H}(\alpha_{\bar{q}}),
  \end{align*}
  where we used Eq.~\eqref{eq:hl-pairing} and the definition of $\bar{H}$ in Eq.~\eqref{eq:barH}.
\end{proof}

Furthermore, as shown in Theorem~\ref{thm:barD} of Appendix~\ref{sec:RedWDCS}~(see also \citet{Ko1992}, \citet{BaSn1993}, \citet{CaLeMaDi1999}, \citet{HoGa2009}), we have the reduced system
\begin{equation}
  \label{eq:ReducedWDCS}
  i_{\bar{X}} \bar{\Omega}^{\rm nh} = d\bar{H}
\end{equation}
on $T^{*}\bar{Q}$ defined with the reduced Hamiltonian $\bar{H}$ and the almost symplectic form
\begin{equation}
  \label{eq:barOmega^nh}
  \bar{\Omega}^{\rm nh} \defeq \bar{\Omega} - \Xi,
\end{equation}
where $\Xi$ is the non-closed two-form on $T^{*}\bar{Q}$ defined in Eq.~\eqref{eq:Xi-def}.

\subsection{Hamilton--Jacobi Theorem for Weakly Degenerate Chaplygin Systems}
The previous subsection showed that a weakly degenerate Chaplygin
system reduces to a non-degenerate Lagrangian and hence an
almost Hamiltonian system~\eqref{eq:ReducedWDCS}. Moreover,
Lemma~\ref{lem:mathcalE-barH} shows how the generalized energy
$\mathcal{E}$ is related to the reduced Hamiltonian $\bar{H}$; see
also the upper half of the diagram~\eqref{dia:Upsilon-bargamma} below.
The lower half of the diagram suggests the relationship between the reduced and original Dirac--Hamilton--Jacobi equations alluded above: Specifically, $\bar{\gamma}$ is a one-form on $\bar{Q} \defeq Q/G$ and is a solution of the reduced Dirac--Hamilton--Jacobi equation~\eqref{eq:ReducedNHHJ} defined below, and the diagram suggests how to define the map $\Upsilon: Q \to \mathcal{K}$ so that it is a solution of the original Dirac--Hamilton--Jacobi equation~\eqref{eq:DiracHJ}.
\begin{equation}
  \label{dia:Upsilon-bargamma}
  \vcenter{
    \xymatrix@!0@R=0.6in@C=0.51975in{
      & \R &
      \\
      \mathcal{K} \ar[ru]^{\mathcal{E}\!\!} & & T^{*}\bar{Q} \ar[ul]_{\!\!\bar{H}} \ar[ll]^{\hl^{\mathcal{K}}}
      \\
      Q \ar[rr]_{\pi} \ar@{-->}[u]^{\Upsilon} & & \bar{Q} \ar[u]_{\bar{\gamma}}
    }
  }
\end{equation}
The whole diagram~\eqref{dia:Upsilon-bargamma} leads us to the following main result of this section:
\begin{theorem}[Reduced Dirac--Hamilton--Jacobi Equation]
  Consider a weakly degenerate Chaplygin system on a connected configuration space $Q$ and assume that the distribution $\Delta_{Q}$ is completely nonholonomic.
  Let $\bar{\gamma}: \bar{Q} \to T^{*}\bar{Q}$ be a one-form on $\bar{Q}$ that satisfies the {\em reduced Dirac--Hamilton--Jacobi equation}
  \begin{equation}
    \label{eq:ReducedNHHJ}
    \bar{H} \comp \bar{\gamma}(\bar{q}) = E,
  \end{equation}
  with a constant $E$, as well as
  \begin{equation}
    \label{eq:dbargamma}
    d\bar{\gamma} + \bar{\gamma}^{*}\Xi = 0,
  \end{equation}
  where $\Xi$ is the two-form on $T^{*}\bar{Q}$ that appeared in the definition of the almost symplectic form $\bar{\Omega}^{\rm nh}$ in Eq.~\eqref{eq:barOmega^nh} (see also Eq.~\eqref{eq:Xi-def}).
  Define $\Upsilon = \mathcal{X} \oplus \gamma : Q \to \mathcal{K}$ by (see the diagram~\eqref{dia:Upsilon-bargamma})
  \begin{equation}
    \label{eq:Upsilon-bargamma}
    \Upsilon(q) \defeq \hl^{\mathcal{K}}_{q} \comp \bar{\gamma} \comp \pi(q)
    = \hl^{\mathcal{K}}_{q}\parentheses{ \bar{\gamma}(\bar{q}) },
  \end{equation}
  where $\bar{q} \defeq \pi(q)$, i.e.,
  \begin{equation*}
    \mathcal{X}(q) \defeq \hl^{\Delta}_{q} \comp (\F\bar{L})^{-1}_{\bar{q}}(\bar{\gamma}(\bar{q})),
    \qquad
    \gamma(q) \defeq \hl^{P}_{q}(\bar{\gamma}(\bar{q})).
  \end{equation*}
  Then, $\Upsilon = \mathcal{X} \oplus \gamma$ satisfies the Dirac--Hamilton--Jacobi equation~\eqref{eq:DiracHJ-2} as well as condition~\eqref{eq:dgamma}.
\end{theorem}

\begin{proof}
  This proof is very similar to that of Theorem~4.1 in \citet{OhFeBlZe2011}.

  The diagram \eqref{dia:Upsilon-bargamma} shows that if the one-form $\bar{\gamma}$ satisfies Eq.~\eqref{eq:ReducedNHHJ} then the map $\Upsilon$ defined by Eq.~\eqref{eq:Upsilon-bargamma} satisfies the Dirac--Hamilton--Jacobi equation~\eqref{eq:DiracHJ-2}.

  To show that it also satisfies the condition~\eqref{eq:dgamma}, we perform the following calculations:
  Let $Y^{\rm h}, Z^{\rm h} \in \mathfrak{X}(Q)$ be arbitrary horizontal vector fields, i.e., $Y^{\rm h}_{q}, Z^{\rm h}_{q} \in \Delta_{Q}(q)$ for any $q \in Q$.
  We start from the following identity:
  \begin{equation}
    \label{eq:gamma-identity}
    d\gamma(Y^{\rm h}, Z^{\rm h})
    = Y^{\rm h}[\gamma(Z^{\rm h})] - Z^{\rm h}[\gamma(Y^{\rm h})] - \gamma([Y^{\rm h},Z^{\rm h}]).
  \end{equation}
  The goal is to show that the right-hand side vanishes.
  Let us first evaluate the first two terms on the right-hand side of the above identity at an arbitrary point $q \in Q$:
  Let $Z_{\bar{q}} \defeq T_{q}\pi_{Q}(Z^{\rm h}_{q}) \in T_{\bar{q}}\bar{Q}$, then $Z^{\rm h}_{q} = \hl^{\Delta}_{q}(Z_{\bar{q}})$.
  Thus, we have
  \begin{align*}
    \gamma(Z^{\rm h})(q)
    &= \ip{ \hl^{P}_{q} \comp \bar{\gamma}(\bar{q}) }{ \hl^{\Delta}_{q}(Z_{\bar{q}}) }
    \\
    &= \ip{ \bar{\gamma}(\bar{q}) }{ Z_{\bar{q}} }
    \\
    &= \bar{\gamma}(Z)(\bar{q}).
  \end{align*}
  Hence, writing $\gamma_{Z} = \bar{\gamma}(Z)$ for short, we have $\gamma(Z^{\rm h}) = \gamma_{Z} \comp \pi$.
  Therefore, defining $Y_{\bar{q}} \defeq T_{q}\pi(Y^{\rm h}_{q})$, i.e., $Y^{\rm h}_{q} = \hl^{\Delta}_{q}(Y_{\bar{q}})$,
  \begin{align*}
    Y^{\rm h}[ \gamma(Z^{\rm h}) ](q)
    &= Y^{\rm h}[ \gamma_{Z} \comp \pi ](q)
    \\
    &= \ip{ d(\gamma_{Z} \comp \pi)_{q} }{ Y^{\rm h}_{q} }
    \\
    &= \ip{ (\pi^{*} d\gamma_{Z})_{q} }{ Y^{\rm h}_{q} }
    \\
    &= \ip{ d\gamma_{Z}(\bar{q}) }{ T_{q}\pi(Y^{\rm h}_{q}) }
    \\
    &= \ip{ d\gamma_{Z}(\bar{q}) }{ Y_{\bar{q}} }
    \\
    &= Y[\gamma_{Z}](\bar{q})
    \\
    &= Y[\bar{\gamma}(Z)](\bar{q}).
  \end{align*}
  Hence, we have
  \begin{align}
     Y^{\rm h}[ \gamma(Z^{\rm h}) ] -  Z^{\rm h}[ \gamma(Y^{\rm h}) ]
     = Y[\bar{\gamma}(Z)] - Z[\bar{\gamma}(Y)],
    \label{eq:YgammaZ-ZgammaY}
  \end{align}
  where we have omitted $q$ and $\bar{q}$ for simplicity.

  Now, let us evaluate the last term on the right-hand side of Eq.~\eqref{eq:gamma-identity}:
  First we would like to decompose $[Y^{\rm h},Z^{\rm h}]_{q}$ into the horizontal and vertical parts.
  Since both $Y^{\rm h}$ and $Z^{\rm h}$ are horizontal, we have\footnote{See, e.g., \citet{KoNo1963}[Proposition~1.3~(3) on p.~65].}
  \begin{equation*}
    \hor([Y^{\rm h},Z^{\rm h}]_{q}) = \hl^{\Delta}_{q}([Y, Z]_{\bar{q}}),
  \end{equation*}
  whereas the vertical part is
  \begin{equation*}
    \ver([Y^{\rm h},Z^{\rm h}]_{q}) = \parentheses{ \mathcal{A}_{q}([Y^{\rm h},Z^{\rm h}]_{q}) }_{Q}(q) = -\parentheses{ \mathcal{B}_{q}(Y^{\rm h}_{q},Z^{\rm h}_{q}) }_{Q}(q),
  \end{equation*}
  where we used the following relation between the connection $\mathcal{A}$ and its curvature $\mathcal{B}$ that holds for horizontal vector fields $Y^{\rm h}$ and $Z^{\rm h}$:
  \begin{align*}
    \mathcal{B}_{q}(Y^{\rm h}_{q}, Z^{\rm h}_{q}) &= d\mathcal{A}_{q}(Y^{\rm h}_{q}, Z^{\rm h}_{q})
    \\
    &= Y^{\rm h}[ \mathcal{A}(Z^{\rm h}) ](q) - Z^{\rm h}[ \mathcal{A}
    (Y^{\rm h}) ](q) - \mathcal{A}([Y^{\rm h},Z^{\rm h}])(q)
    \\
    &= -\mathcal{A}([Y^{\rm h},Z^{\rm h}])(q).
  \end{align*}
  As a result, we have the decomposition
  \begin{equation*}
    [Y^{\rm h},Z^{\rm h}]_{q} = \hl^{\Delta}_{q}([Y, Z]_{\bar{q}}) - \parentheses{ \mathcal{B}_{q}(Y^{\rm h}_{q},Z^{\rm h}_{q}) }_{Q}(q).
  \end{equation*}
  Therefore,
  \begin{align}
    \gamma([Y^{\rm h},Z^{\rm h}])(q)
    &= \ip{ \hl^{P}_{q} \comp \bar{\gamma} \comp \pi(q) }{ \hl^{\Delta}_{q}([Y, Z]_{\bar{q}}) }
    - \ip{ \hl^{P}_{q} \comp \bar{\gamma} \comp \pi(q) }{ \parentheses{ \mathcal{B}_{q}(Y^{\rm h}_{q},Z^{\rm h}_{q}) }_{Q}(q) }
    \nonumber\\
    &= \ip{ \bar{\gamma}(\bar{q}) }{ [Y, Z]_{\bar{q}} }
    - \ip{ {\bf J}\parentheses{ \hl^{P}_{q} \comp \bar{\gamma}(\bar{q}) } }{ \mathcal{B}_{q}(Y^{\rm h}_{q},Z^{\rm h}_{q}) }
    \nonumber\\
    &= \bar{\gamma}([Y,Z])(\bar{q})
    - \bar{\gamma}^{*}\Xi(Y, Z)(\bar{q}),
    \label{eq:gammaYZ}
  \end{align}
  where the second equality follows from Eq.~\eqref{eq:hl-pairing} and the definition of the momentum map ${\bf J}$; the last equality follows from the definition of $\Xi$ in Eq.~\eqref{eq:Xi-def}:
  Let $\pi_{\bar{Q}}: T^{*}\bar{Q} \to \bar{Q}$ be the cotangent bundle projection; then we have
  \begin{align*}
    \bar{\gamma}^{*}\Xi(Y, Z)(\bar{q})
    &= \Xi_{\bar{\gamma}(\bar{q})} \parentheses{ T\bar{\gamma}(Y_{\bar{q}}), T\bar{\gamma}(Z_{\bar{q}}) }
    \\
    &= \ip{ {\bf J} \comp \hl^{P}_{q} \parentheses{ \bar{\gamma}(\bar{q})} }{ \mathcal{B}_{q}\parentheses{\hl^{\Delta}_{q}(Y_{\bar{q}}), \hl^{\Delta}_{q}(Z_{\bar{q}})} },
  \end{align*}
  since $\pi_{\bar{Q}} \comp \bar{\gamma} = \id_{\bar{Q}}$ and thus $T\pi_{\bar{Q}} \comp T\bar{\gamma} = \id_{T\bar{Q}}$.
  Substituting Eqs.~\eqref{eq:YgammaZ-ZgammaY} and \eqref{eq:gammaYZ} into Eq.~\eqref{eq:gamma-identity}, we obtain
  \begin{align*}
    d\gamma(Y^{\rm h}, Z^{\rm h})
    &= Y[\bar{\gamma}(Z)] - Z[\bar{\gamma}(Y)] - \bar{\gamma}([Y,Z])(\bar{q}) + \bar{\gamma}^{*}\Xi(Y, Z)
    \\
    &= d\bar{\gamma}(Y, Z) + \bar{\gamma}^{*}\Xi(Y, Z)
    \\
    &= \parentheses{ d\bar{\gamma} + \bar{\gamma}^{*}\Xi }(Y, Z)
    = 0. \qedhere
  \end{align*}
\end{proof}

\begin{example}[Simplified Roller Racer; see Examples~\ref{ex:RollerRacer} and \ref{ex:RollerRacer2}]
  The Lie algebra $\mathfrak{g}$ of $G = \R^{2}$ is identified with $\R^{2}$; let be $(\xi, \eta)$ the coordinates for $\mathfrak{g}$ such that $\xi_{Q} = \tpd{}{x}$ and $\eta_{Q} = \tpd{}{y}$.
  Then, we may write the connection $\mathcal{A}: TQ \to \mathfrak{g}$ as
  \begin{equation*}
    \mathcal{A} = \omega^{1} \otimes \pd{}{\xi} + \omega^{2} \otimes \pd{}{\eta},
  \end{equation*}
  where $\omega^{1}$ and $\omega^{2}$ are the constraint one-forms defined in Eq.~\eqref{eq:omegas-RollerRacer}; hence its curvature is given by
  \begin{equation*}
    \mathcal{B} = -\csc^{2}\phi [d_{1} \cos\theta + d_{2} \cos(\theta + \phi)] d\theta \wedge d\phi \otimes d\xi
     - \csc^{2}\phi [d_{1} \sin\theta + d_{2}\sin(\theta+\phi)] d\theta \wedge d\phi \otimes d\eta.
  \end{equation*}
  Furthermore, the momentum map ${\bf J}: T^{*}Q \to \mathfrak{g}^{*}$ is given by
  \begin{equation*}
    {\bf J}(p_{q}) = p_{x}\,d\xi + p_{y}\,d\eta.
  \end{equation*}
  Therefore, we have
  \begin{equation*}
    \Xi = -p_{\phi} \parentheses{ \frac{d_{1}}{d_{2}} + \cos\phi } \csc\phi\, d\theta \wedge d\phi.
  \end{equation*}

  Since the reduced Lagrangian $\bar{L}$ (see Eq.~\eqref{eq:barL-RollerRacer}) is non-degenerate, we have the reduced Hamiltonian $\bar{H}: T^{*}\bar{Q} \to \R$ given by
  \begin{equation*}
    \bar{H} = \frac{1}{2I_{1}}\brackets{ p_{\theta} - \parentheses{ 1 + \frac{d_{1}}{d_{2}} \cos\phi }p_{\phi}\, }^{2}
    + \frac{\sin^{2}\phi}{2 m_{1}d_{2}^{2}}\,p_{\phi}^{2}.
  \end{equation*}

  We assume the ansatz
  \begin{equation*}
    \bar{\gamma}_{\phi}(\theta, \phi) = \bar{\gamma}_{\phi}(\phi).
  \end{equation*}
  Then, the reduced Dirac--Hamilton--Jacobi equation \eqref{eq:ReducedNHHJ} gives
  \begin{equation*}
    \frac{1}{2 I_{1}}\brackets{
      \bar{\gamma}_{\theta}(\theta,\phi) - \parentheses{ 1 + \frac{d_{1}}{d_{2}} \cos\phi }\bar{\gamma}_{\phi}(\phi)
    }^{2}
    + \frac{\sin^{2}\phi}{2 m_{1}d_{2}^{2}}\, \bar{\gamma}_{\phi}(\phi)^{2} = E,
  \end{equation*}
  which implies that $\bar{\gamma}_{\theta}(\theta,\phi) = \bar{\gamma}_{\theta}(\phi)$.
  Solving this for $\bar{\gamma}_{\theta}(\phi)$ gives
  \begin{equation*}
    \bar{\gamma}_{\theta}(\phi) =
    \parentheses{ 1 + \frac{d_{1}}{d_{2}}\,\cos\phi } \bar{\gamma}_{\phi}(\phi)
    \pm \sqrt{ I_{1} \parentheses{ 2E - \frac{\sin^{2}\phi}{m_{1}d_{2}^{2}}\,\bar{\gamma}_{\phi}(\phi)^{2} } }.
  \end{equation*}
  Substituting the first case into Eq.~\eqref{eq:dbargamma}, we obtain
  \begin{equation*}
    \bar{\gamma}_{\phi}'(\phi) = -\cot\phi\,\bar{\gamma}_{\phi}(\phi),
  \end{equation*}
  which gives
  \begin{equation*}
    \bar{\gamma}_{\phi}(\phi) = C \csc\phi
  \end{equation*}
  for some constant $C$.
  Therefore,
  \begin{equation*}
    \bar{\gamma}_{\theta}(\phi) = C \parentheses{ 1 + \frac{d_{1}}{d_{2}}\,\cos\phi } \csc\phi
    + \sqrt{ I_{1} \parentheses{ 2E - \frac{C^{2}}{m_{1}d_{2}^{2}} } }.
  \end{equation*}
  It is straightforward to check that, with the choice
  \begin{equation*}
    C = d_{2} \sqrt{ m_{1}(2E - I_{1} v_{\theta}^{2}) },
  \end{equation*}
  Eq.~\eqref{eq:Upsilon-bargamma} gives the solution obtained in Example~\ref{ex:RollerRacer2}.
\end{example}

\begin{remark}
  Notice that the ansatz we used here is less elaborate compared to the one, Eq.~\eqref{eq:ansatz1-RollerRacer}, used for the Dirac--Hamilton--Jacobi equation without the reduction.
  Specifically, accounting for the $\R^{2}$-symmetry is not necessary here, since the reduced Dirac--Hamilton--Jacobi equation is defined for the $\R^{2}$-reduced system.
\end{remark}

\section{Conclusion and Future Work}
\subsection*{Conclusion}
We developed Hamilton--Jacobi theory for degenerate Lagrangian
systems with holonomic and nonholonomic constraints.
In particular, we illustrated, through a few examples, that solutions
of the Dirac--Hamilton--Jacobi equation can be used to obtain
exact solutions of the equations of motion. Also, motivated by
those degenerate Lagrangian systems that appear as simplified
models of nonholonomic mechanical systems, we introduced a class of
degenerate nonholonomic Lagrangian systems that reduce to
non-degenerate almost Hamiltonian systems. We then showed that the
Dirac--Hamilton--Jacobi equation reduces to the nonholonomic
Hamilton--Jacobi equation for the reduced non-degenerate system.
\subsection*{Future Work}
\begin{itemize}
\item {\em Relationship with discrete variational Dirac mechanics}.
  Hamilton--Jacobi theory has been an important ingredient in discrete mechanics and symplectic integrators from both the theoretical and implementation points of view~(see \citet{MaWe2001}[Sections~1.7, 1.8, 4.7, and 4.8] and \citet{ChSc1990}).
  It is interesting to see if the Dirac--Hamilton--Jacobi equation plays the same role in discrete variational Dirac mechanics of \citet{LeOh2010, LeOh2011}.
  \smallskip
\item {\em Hamilton--Jacobi theory for systems with Lagrangians linear in velocity}.
  As briefly mentioned in Section~\ref{ssec:Limitation}, the Dirac--Hamilton--Jacobi equation is not appropriate for those systems with Lagrangians that are linear in velocity.
  However, \citet{RoSc2003}(Example~4) illustrate that their formulation of the Hamilton--Jacobi equation can be applied to such systems.
  We are interested in a possible generalization of our formulation to deal with such systems, and also a link with their formulation.
  \smallskip
\item {\em Asymptotic analysis of massless approximation}.
  Massless approximations for some nonholonomic systems seem to give good approximations to the full formulation.
  It seems that the nonholonomic constraints ``regularize'' the otherwise singular perturbation problem, and hence makes the massless approximations viable.
  We expect that asymptotic analysis will reveal how the perturbation problem becomes regular, particularly for those cases where massless approximations lead to weakly degenerate Chaplygin systems.
  \smallskip
\item {\em Hamilton--Jacobi theory for general systems on the Pontryagin bundle}.
  Section~\ref{ssec:LDSonPontryaginBundle} naturally leads us to consider systems on the Pontryagin bundle described by an arbitrary Dirac structure.
  We are interested in this generalization, its corresponding Hamilton--Jacobi theory, and its applications.
\end{itemize}

\begin{acknowledgments}
We would like to thank Anthony Bloch, Henry Jacobs, Jerrold Marsden, Joris Vankerschaver, and Hiroaki Yoshimura for their helpful comments, and also Anthony Bloch and Wang Sang Koon for their permission to use their figures.
This material is based upon work supported by the National Science Foundation under the applied mathematics grant DMS-0726263, the Faculty Early Career Development (CAREER) award DMS-1010687, the FRG grant DMS-1065972, MICINN (Spain) grants MTM2009-13383 and MTM2009-08166-E, and the projects of the Canary government SOLSUBC200801000238 and ProID20100210.
\end{acknowledgments}

\appendix

\section{Reduction of Weakly Degenerate Chaplygin Systems}
\label{sec:RedWDCS}
\subsection{Constrained Dirac Structure}
We may restrict the Dirac structure $D_{\Delta_{Q}}$ to $P \subset
T^{*}Q$ as follows~(see \citet{YoMa2006b}[Section~5.6] and
references therein): Let us define a distribution $\mathcal{H}
\subset TP$ on $P$ by
\begin{equation}
  \label{eq:mathcalH}
  \mathcal{H} \defeq TP \cap \Delta_{T^{*}Q},
\end{equation}
and also, using the inclusion $\iota_{P}: P \hookrightarrow
T^{*}Q$, define the two-form $\Omega_{P} \defeq \iota_{P}^{*}
\Omega$ on $P$. Then, define the {\em constrained Dirac structure}
$D_{P} \subset TP \oplus T^{*}P$, for each $z \in P$,  by
\begin{equation*}
  D_{P}(z) \defeq \setdef{
    (v_{z},\alpha_{z}) \in T_{z}P \oplus T^{*}_{z}P
  }{
    v_{z} \in \mathcal{H}_{z},\;
    \alpha_{z} - \Omega_{P}^{\flat}(z)(v_{z}) \in \mathcal{H}^\circ_z
  },
\end{equation*}
where $\Omega_{P}^{\flat}: TP \to T^{*}P$ is the flat map induced by $\Omega_{P}$. Then, we have the {\em constrained
Lagrange--Dirac system} defined by
\begin{equation}
  \label{eq:ConstrainedLDS}
  (X_{P}, \mathfrak{D}L_{\rm c}) \in D_{P},
\end{equation}
where $X_{P}$ is a vector field on $P$, $L_{\rm c} \defeq
L|_{\Delta_{Q}}$ the constrained Lagrangian, and
$\mathfrak{D}L_{\rm c}(u) \defeq \mathfrak{D}L(u)|_{TP}$ for any
$u \in \Delta_{Q}$.

If the constrained Lagrangian $L_{\rm c}$ is non-degenerate, i.e.,
the partial Legendre transformation $\Leg|_{\Delta_{Q}}: \Delta_{Q}
\to P$ is invertible, then we may define the constrained
Hamiltonian~\cite{YoMa2007b} $H_{P}: P \to \R$ by
\begin{equation*}
  H_{P}(p_{q}) \defeq \ip{ p_{q} }{ v_{q} } - L_{\rm c}(v_{q}),
\end{equation*}
where $v_{q} \defeq (\Leg|_{\Delta_{Q}})^{-1}(p_{q})$. Then, the
constrained Lagrange--Dirac system~\eqref{eq:ConstrainedLDS},
is equivalent to the {\em constrained implicit Hamiltonian system} defined by
\begin{equation}
  \label{eq:ConstrainedIHS}
  (X_{P}, dH_{P}) \in D_{P}.
\end{equation}
\begin{remark}
  Let
  \begin{equation}
    \label{eq:Omega_mathcalH}
    \Omega_{\mathcal{H}} \defeq \Omega_{P}|_{\mathcal{H}}
  \end{equation}
  be the restriction of $\Omega_{P}$ to $\mathcal{H} \subset TP$ and hence a skew-symmetric bilinear form in $\mathcal{H}$.
  If $\Omega_{\mathcal{H}}$ is non-degenerate, then Eq.~\eqref{eq:ConstrainedIHS} gives
  \begin{equation*}
    i_{X_{P}} \Omega_{\mathcal{H}} = dH_{P}|_{\mathcal{H}},
  \end{equation*}
  which is nonholonomic Hamilton's equations of \citet{BaSn1993} (see also \citet{KoMa1997c}).
\end{remark}

\subsection{Reduction of Constrained Dirac Structure}
Let us now show how to reduce the constrained Dirac structure $D_{P}$ to a Dirac structure on $T^{*}\bar{Q}$, where $\bar{Q} \defeq Q/G$.
This special case of Dirac reduction to follow gives a Dirac point of view on the nonholonomic reduction of \citet{Ko1992}, and hence provides a natural framework for the reduction of weakly degenerate Chaplygin systems.
See \citet{YoMa2009} for reduction of Dirac structures without constraints, \citet{JoRa2011} for the relationship between Dirac and nonholonomic reduction of \citet{BaSn1993}; see also \citet{CaCaCrIb1986, CaLeMaDi1999} for a theory of reducing degenerate Lagrangian systems to non-degenerate ones.

Let $\Phi: G \times Q \to Q$ be the action of the Lie group $G$
given in Definition~\ref{def:ChaplyginSystems} and
$T^{*}\Phi_{g^{-1}}: T^{*}Q \to T^{*}Q$ be its cotangent lift
defined by
\begin{equation*}
  \ip{T^{*}\Phi_{g^{-1}}(\alpha)}{v}
  = \ip{\alpha}{T\Phi_{g^{-1}}(v)}.
\end{equation*}
It is easy to show that the $G$-symmetries of the Lagrangian $L$
and the distribution $\Delta_{Q}$ imply that the submanifold $P
\subset T^{*}Q$ is invariant under the action of the cotangent
lift.
Hence, we may restrict the action to $P$ and define $\Phi^{P}: G \times P \to P$, i.e., $\Phi^{P}_{g}: P \to P$ by $\Phi^{P}_{g} \defeq T^{*}\Phi_{g^{-1}}|_{P}$ for any $g \in G$.
This gives rise to the principal bundle
\begin{equation*}
  \pi^{P}_{G}: P \to P/G.
\end{equation*}
The geometric structure of weakly degenerate Chaplygin systems
summarized in Section~\ref{ssec:GeometryOfWDCS} gives rise to a
diffeomorphism $\varphi: T^{*}\bar{Q} \to P/G$; this then induces
the map $\rho: P \to T^{*}\bar{Q}$ so that the diagram below
commutes (see \citet{HoGa2009}).
\begin{equation}
  \label{dia:rho}
  \begin{array}{c}
    \xymatrix@!0@R=0.8in@C=0.9in{
      P \ar[d]_{\pi^{P}_{G}} \ar[dr]^{\rho}&
      \\
      P/G \ar[r]_{\varphi^{-1}} & T^{*}\bar{Q}
    }
  \end{array}
\end{equation}

Furthermore, the principal connection $\mathcal{A}: TQ \to
\mathfrak{g}$ defined in Eq.~\eqref{eq:mathcalA} induces the
principal connection $\mathcal{A}_{P}: TP \to \mathfrak{g}$
defined by
\begin{equation*}
  \mathcal{A}_{P} \defeq (\pi_{Q} \comp \iota_{P})^{*} \mathcal{A},
\end{equation*}
and the horizontal space for this principal connection is
$\mathcal{H} \subset TP$ defined in Eq.~\eqref{eq:mathcalH}, i.e.,
$\mathcal{H} = \ker\mathcal{A}_{P}$~\cite{HoGa2009}.
Therefore, writing $[z] \defeq \pi^{P}_{G}(z) \in P/G$, we have the horizontal lift
\begin{equation*}
  \hl^{\mathcal{H}}_{z}: T_{[z]}(P/G) \to \mathcal{H}_{z};
  \quad
  v_{[z]} \mapsto (T_{z}\pi^{P}_{G}|_{\mathcal{H}_{z}})^{-1}(v_{[z]}).
\end{equation*}
Then, clearly the following diagram commutes:
\begin{equation}
  \label{dia:Trho}
  \begin{array}{c}
    \xymatrix@!0@R=0.9in@C=1.15in{
      T_{z}P \ar[dr]^{T_{z}\rho}&
      \\
      T_{[z]}(P/G) \ar[r]_{\ T_{[z]}\varphi^{-1}} \ar[u]^{\hl^{\mathcal{H}}_{z}} & T_{\bar{z}}T^{*}\bar{Q}
    }
  \end{array}
\end{equation}
where $\bar{z} \defeq \varphi^{-1}([z]) \in T^{*}\bar{Q}$.

\begin{lemma}
  The two-form $\Omega_{P}$ is invariant under the $G$-action, i.e., for any $g \in G$,
  \begin{equation}
    \label{eq:Omega_P-preservation}
    (\Phi^{P}_{g})^{*} \Omega_{P} = \Omega_{P}.
  \end{equation}
\end{lemma}
\begin{proof}
  Using the relation $T^{*}\Phi_{g^{-1}} \comp \iota_{P} = \iota_{P} \comp \Phi^{P}_{g}$, we have
  \begin{align*}
    (\Phi^{P}_{g})^{*} \Omega_{P} &= (\Phi^{P}_{g})^{*} \iota_{P}^{*} \Omega
    \\
    &= (\iota_{P} \comp \Phi^{P}_{g})^{*} \Omega
    \\
    &= (T^{*}\Phi_{g^{-1}} \comp \iota_{P})^{*} \Omega
    \\
    &= (\iota_{P})^{*} \comp (T^{*}\Phi_{g^{-1}})^{*} \Omega
    \\
    &= \iota_{P}^{*} \Omega
    \\
    &= \Omega_{P},
  \end{align*}
  where we used the fact that the cotangent lift $T^{*}\Phi_{g^{-1}}$ is symplectic.
\end{proof}
Now, consider the action of $G$ on the Whitney sum $TP \oplus
T^{*}P$ defined by
\begin{equation*}
  \Psi: G \times (TP \oplus T^{*}P) \to TP \oplus T^{*}P;
  \quad
  (g, (v_{z}, \alpha_{z})) \mapsto \parentheses{ T_{z}\Phi^{P}_{g}(v_{z}), T_{g z}^{*}\Phi^{P}_{g^{-1}}(\alpha_{z}) } \eqdef \parentheses{ g \cdot v_{z},  g \cdot \alpha_{z} }.
\end{equation*}
Then, we have the following:
\begin{proposition}
  The constrained Dirac structure $D_{P}$ is invariant under the action $\Psi$ defined above.
\end{proposition}

\begin{proof}
  Let $z \in P$ be arbitrary and $(v_{z}, \alpha_{z}) \in D_{P}(z)$.
  Then, $v_{z} \in \mathcal{H}_{z}$ and $\alpha_{z} - \Omega_{P}^{\flat}(v_{z}) \in \mathcal{H}_{z}^{\circ}$.
  Now, the $G$-invariance of $\mathcal{H} = \ker\mathcal{A}_{P}$ implies $T\Phi_{g}(v_{z}) \in \mathcal{H}_{g z}$.
  Also, for any $w_{g z} \in \mathcal{H}_{g z}$ we have $w_{z} \defeq T_{gz}\Phi^{P}_{g^{-1}}(w_{z}) \in \mathcal{H}_{z}$, and thus
  \begin{align*}
    \ip{ T_{g z}^{*}\Phi^{P}_{g^{-1}}(\alpha_{z}) - \Omega_{P}^{\flat}\parentheses{ T_{z}\Phi^{P}_{g}(v_{z}) } }{ w_{g z} }
    &= \ip{ \alpha_{z} }{ T_{g z}\Phi^{P}_{g^{-1}}(w_{g z}) }
    - \Omega_{P}\parentheses{ T_{z}\Phi^{P}_{g}(v_{z}), T_{z}\Phi^{P}_{g}(w_{z}) }
    \\
    &= \ip{ \alpha_{z} }{ w_{z} }
    - \Omega_{P}\parentheses{ T_{z}\Phi^{P}_{g}(v_{z}), T_{z}\Phi^{P}_{g}(w_{z}) }
    \\
    &= \ip{ \alpha_{z} }{ w_{z} }
    - (\Phi^{P}_{g})^{*} \Omega_{P}\parentheses{ v_{z}, w_{z} }
    \\
    &= \ip{ \alpha_{z} }{ w_{z} }
    - \Omega_{P}\parentheses{ v_{z}, w_{z} }
    \\
    &= \ip{ \alpha_{z} - \Omega_{P}^{\flat}(v_{z}) }{ w_{z} }
    \\
    &= 0,
  \end{align*}
  where the fourth line follows from Eq.~\eqref{eq:Omega_P-preservation}.
  Hence
  \begin{equation*}
    (g \cdot v_{z}, g \cdot \alpha_{z}) = \parentheses{ T_{z}\Phi^{P}_{g}(v_{z}), T_{g z}^{*}\Phi^{P}_{g^{-1}}(\alpha_{z}) } \in D_{P}(g z),
  \end{equation*}
  and thus the claim follows.
\end{proof}

Now, the main result in this section is the following:
\begin{theorem}
  \label{thm:barD}
  The reduced constrained Dirac structure $[D_{P}]_{G} \defeq D_{P}/G$ is identified with the Dirac structure $\bar{D}$ on $T^{*}\bar{Q}$ defined, for any $\bar{z} \in T^{*}\bar{Q}$, by
  \begin{equation}
    \label{eq:barD}
    \bar{D}(\bar{z}) \defeq \setdef{
      (v_{\bar{z}},\alpha_{\bar{z}}) \in T_{\bar{z}}T^{*}\bar{Q} \oplus T^{*}_{\bar{z}}T^{*}\bar{Q}
    }{
      \alpha_{\bar{z}} = (\bar{\Omega}^{\rm nh})^{\flat}(v_{\bar{z}})
    },
  \end{equation}
  where $\bar{\Omega}^{\rm nh} = \bar{\Omega} - \Xi$ with $\bar{\Omega}$ being the standard symplectic form on $T^{*}\bar{Q}$, and the two-form $\Xi$ on $T^{*}\bar{Q}$ is defined as follows:
  For any $\alpha_{\bar{q}} \in T_{\bar{q}}^{*}\bar{Q}$ and $\mathcal{Y}_{\alpha_{\bar{q}}}, \mathcal{Z}_{\alpha_{\bar{q}}} \in T_{\alpha_{\bar{q}}}T^{*}\bar{Q}$, let $Y_{\bar{q}} \defeq T\pi_{\bar{Q}}(\mathcal{Y}_{\alpha_{\bar{q}}})$ and $Z_{\bar{q}} \defeq T\pi_{\bar{Q}}(\mathcal{Z}_{\alpha_{\bar{q}}})$ where $\pi_{\bar{Q}}: T^{*}\bar{Q} \to \bar{Q}$ is the cotangent bundle projection, and then set
  \begin{equation}
    \label{eq:Xi-def}
    \Xi_{\alpha_{\bar{q}}}(\mathcal{Y}_{\alpha_{\bar{q}}}, \mathcal{Z}_{\alpha_{\bar{q}}})
    \defeq \ip{{\bf J} \comp \hl^{P}_{q} (\alpha_{\bar{q}})}{ \mathcal{B}_{q}\!\parentheses{ \hl^{\Delta}_{q}(Y_{\bar{q}}), \hl^{\Delta}_{q}(Z_{\bar{q}}) } },
  \end{equation}
  where ${\bf J}: T^{*}Q \to \mathfrak{g}^{*}$ is the momentum map corresponding to the $G$-action, and $\mathcal{B}$ is the curvature two-form of the connection $\mathcal{A}$.
\end{theorem}

\begin{lemma}
  \label{lem:f}
  Define, for any $z \in P$,
  \begin{equation*}
    f_{z}: T_{z}P \oplus T^{*}_{z}P \to T_{\bar{z}}T^{*}\bar{Q} \oplus T^{*}_{\bar{z}}T^{*}\bar{Q};
    \qquad
    f_{z}(v_{z}, \alpha_{z}) = \parentheses{ T_{z}\rho(v_{z}), T^{*}_{\bar{z}}\varphi \comp (\hl^{\mathcal{H}}_{z})^{*}(\alpha_{z}) },
  \end{equation*}
  where $(\hl^{\mathcal{H}}_{z})^{*}: T^{*}_{z}P \to T^{*}_{[z]}(P/G)$ is the adjoint map of $\hl^{\mathcal{H}}_{z}$.
  Then, $f$ is $G$-invariant, i.e., $f \comp \Psi_{g} = f$ for any $g \in G$.
\end{lemma}

\begin{remark}
  The map $f_{z}|_{D_{P}(z)}$, i.e., $f_{z}$ defined above restricted to $D_{P}(z) \subset T_{z}P \oplus T^{*}_{z}P$, is the backward Dirac map (see \citet{BuRa2003}) of
  \begin{equation*}
    \phi_{z} \defeq \hl^{\mathcal{H}}_{z} \comp T_{\bar{z}}\varphi: T_{\bar{z}}T^{*}\bar{Q} \to T_{z}P,
  \end{equation*}
  that is, $f_{z} = \mathcal{B}\phi_{z}$ using the notation in \citet{BuRa2003}; hence the image $f(D_{P}) \subset TT^{*}\bar{Q} \oplus T^{*}T^{*}\bar{Q}$ is a Dirac structure.
\end{remark}

\begin{proof}[Proof of Lemma~\ref{lem:f}]
  Let $(v_{z}, \alpha_{z}) \in T_{z}P \oplus T^{*}_{z}P$ and $(\tilde{v}_{g z}, \tilde{\alpha}_{g z}) \defeq \Psi_{g}( v_{z}, \alpha_{z} )$ for $g \in G$, i.e.,
  \begin{equation*}
    \tilde{v}_{g z} = T_{z}\Phi^{P}_{g}(v_{z}),
    \qquad
    \tilde{\alpha}_{g z} = T_{g z}^{*}\Phi^{P}_{g^{-1}}(\alpha_{z}).
  \end{equation*}
  Using the identities $\rho = \varphi^{-1} \comp \pi^{P}_{G}$ (see diagram~\eqref{dia:rho}) and $\pi^{P}_{G} \comp \Phi^{P}_{g} = \pi^{P}_{G}$, we have
  \begin{align*}
    T_{g z}\rho(\tilde{v}_{g z})
    &= T_{g z}\rho \comp T_{z}\Phi^{P}_{g}(v_{z})
    \\
    &= T_{[z]}\varphi^{-1} \comp T_{g z}\pi^{P}_{G} \comp T_{z}\Phi^{P}_{g}(v_{z})
    \\
    &= T_{[z]}\varphi^{-1} \comp T_{z}(\pi^{P}_{G} \comp \Phi^{P}_{g})(v_{z})
    \\
    &= T_{[z]}\varphi^{-1} \comp T_{z}\pi^{P}_{G}(v_{z})
    \\
    &= T_{z}\rho(v_{z}).
  \end{align*}
  On the other hand, for any $w_{[z]} \in T_{[z]}(P/G$),
  \begin{align*}
    (\hl^{\mathcal{H}}_{g z})^{*}(\tilde{\alpha}_{g z})
    &= (\hl^{\mathcal{H}}_{g z})^{*} \comp T_{g z}^{*}\Phi^{P}_{g^{-1}}(\alpha_{z})
    \\
    &= \parentheses{ T_{g z}\Phi^{P}_{g^{-1}} \comp \hl^{\mathcal{H}}_{g z} }^{*} (\alpha_{z})
    \\
    &= (\hl^{\mathcal{H}}_{z})^{*} (\alpha_{z}),
  \end{align*}
  because of the invariance property of the horizontal lift $\hl^{\mathcal{H}}$, i.e., $T_{g z}\Phi^{P}_{g^{-1}} \comp \hl^{\mathcal{H}}_{g z} = \hl^{\mathcal{H}}_{z}$.
  Hence it follows that $f_{g z} \circ \Psi_{g}(v_{z}, \alpha_{z}) = f_{g z}(\tilde{v}_{g z}, \tilde{\alpha}_{g z}) = f_{z}(v_{z}, \alpha_{z})$.
\end{proof}

\begin{proof}[Proof of Theorem~\ref{thm:barD}]
  Lemma~\ref{lem:f} implies that the map $f|_{D_{P}}$ defined above induces the following well-defined map:
  \begin{equation*}
    \bar{f}: [D_{P}]_{G} \to TT^{*}\bar{Q} \oplus T^{*}T^{*}\bar{Q};
    \quad
    [(v_{z}, \alpha_{z})]_{G} \mapsto \parentheses{ T_{z}\rho(v_{z}), T^{*}_{\bar{z}}\varphi \comp (\hl^{\mathcal{H}}_{z})^{*}(\alpha_{z}) },
  \end{equation*}
  i.e., the diagram below commutes.
  \begin{equation*}
    \begin{array}{c}
      \xymatrix@!0@R=0.8in@C=0.9in{
        D_{P} \ar[d]_{/G} \ar[dr]^{f|_{D_{P}}} &
        \\
        [D_{P}]_{G} \ar[r]_{\bar{f}} & \bar{D}
      }
    \end{array}
  \end{equation*}

  Let us look into the image $\bar{D} \defeq \bar{f}([D_{P}]_{G})$.
  Notice first that
  \begin{equation*}
    T_{z}\rho(\mathcal{H}_{z}) = T_{[z]}\varphi^{-1} \comp T_{z}\pi^{P}_{G}(\mathcal{H}_{z}) = T_{\bar{z}}T^{*}\bar{Q},
  \end{equation*}
  since $T_{z}\pi^{P}_{G}(\mathcal{H}_{z}) = T_{[z]}(P/G)$ and $T_{[z]}\varphi^{-1}$ is surjective.

  On the other hand, notice that $w^{\rm h}_{z} \defeq \hl^{\mathcal{H}}_{z} \comp T_{\bar{z}}\varphi(w_{\bar{z}})$ is in $\mathcal{H}_{z}$ for any $w_{\bar{z}} \in T_{\bar{z}}T^{*}\bar{Q}$, whereas $\alpha_{z} - \Omega_{P}^{\flat}(v_{z}) \in \mathcal{H}^{\circ}_{z}$.
  So we have
  \begin{align*}
    0 &= \ip{ \alpha_{z} - \Omega_{P}^{\flat}(v_{z}) }{ \hl^{\mathcal{H}}_{z} \comp T_{\bar{z}}\varphi(w_{\bar{z}}) }
    \\
    &= \ip{ T^{*}_{\bar{z}}\varphi \comp (\hl^{\mathcal{H}}_{z})^{*}\alpha_{z} - T^{*}_{\bar{z}}\varphi \comp (\hl^{\mathcal{H}}_{z})^{*}\Omega_{P}^{\flat}(v_{z}) }{ w_{\bar{z}} }.
  \end{align*}
  Therefore,
  \begin{equation*}
    T^{*}_{\bar{z}}\varphi \comp (\hl^{\mathcal{H}}_{z})^{*}\alpha_{z} = T^{*}_{\bar{z}}\varphi \comp (\hl^{\mathcal{H}}_{z})^{*}\Omega_{P}^{\flat}(v_{z}).
  \end{equation*}
  However, for an arbitrary $w_{\bar{z}} \in T_{\bar{z}}T^{*}\bar{Q}$,
  \begin{align*}
    \ip{ T^{*}_{\bar{z}}\varphi \comp (\hl^{\mathcal{H}}_{z})^{*}\Omega_{P}^{\flat}(v_{z}) }{ w_{\bar{z}} }
    &= \Omega_{P}\!\parentheses{ v_{z}, \hl^{\mathcal{H}}_{z} \comp T_{\bar{z}}\varphi(w_{\bar{z}}) }
    \\
    &= \Omega_{\mathcal{H}}\!\parentheses{ v_{z}, \hl^{\mathcal{H}}_{z} \comp T_{\bar{z}}\varphi(w_{\bar{z}}) }
    \\
    &= \rho^{*}\bar{\Omega}^{\rm nh}\!\parentheses{ v_{z}, \hl^{\mathcal{H}}_{z} \comp T_{\bar{z}}\varphi(w_{\bar{z}}) }
    \\
    &= \bar{\Omega}^{\rm nh}\!\parentheses{ T_{z}\rho(v_{z}), T_{z}\rho \comp \hl^{\mathcal{H}}_{z} \comp T_{\bar{z}}\varphi(w_{\bar{z}}) }
    \\
    &= \bar{\Omega}^{\rm nh}\!\parentheses{ T_{z}\rho(v_{z}), w_{\bar{z}} }
    \\
    &= \ip{ (\bar{\Omega}^{\rm nh})^{\flat} \comp T_{z}\rho(v_{z}) }{ w_{\bar{z}} },
  \end{align*}
  where the second line follows from the definition of $\Omega_{\mathcal{H}}$, Eq.~\eqref{eq:Omega_mathcalH}, since $(v_{z}, \alpha_{z}) \in D_{P}(z)$ implies $v_{z} \in \mathcal{H}_{z}$; the third line follows from $\rho^{*}\bar{\Omega}^{\rm nh}|_{\mathcal{H}} = \Omega_{\mathcal{H}}$ (see \citet{HoGa2009}[Proposition~2.2]); the fifth from diagram~\eqref{dia:Trho}.
  As a result, we have
  \begin{equation*}
    T^{*}_{\bar{z}}\varphi \comp (\hl^{\mathcal{H}}_{z})^{*}\alpha_{z} = (\bar{\Omega}^{\rm nh})^{\flat} \comp T_{z}\rho(v_{z}),
  \end{equation*}
  and thus
  \begin{equation*}
    \bar{f}\parentheses{ [(v_{z}, \alpha_{z})]_{G} }
    = f(v_{z}, \alpha_{z})
    = \parentheses{ T_{z}\rho(v_{z}), (\bar{\Omega}^{\rm nh})^{\flat} \comp T_{z}\rho(v_{z}) }.
  \end{equation*}
  Since $T_{z}\rho(\mathcal{H}_{z}) = T_{\bar{z}}T^{*}\bar{Q}$, the image $\bar{D} = \bar{f}([D_{P}]_{G}) = f(D_{P})$ is given by Eq.~\eqref{eq:barD}.
\end{proof}

\subsection{Reduction of Weakly Degenerate Chaplygin Systems}
Reduced dynamics of the constrained implicit Hamiltonian system,
Eq.~\eqref{eq:ConstrainedIHS}, for weakly Chaplygin systems
follows easily from Theorem~\ref{thm:barD}: For weakly Chaplygin
systems, it is straightforward to show that the constrained
Hamiltonian $H_{P}$ is related to the reduced Hamiltonian defined
in Eq.~\eqref{eq:barH} as follows:
\begin{equation}
  \label{eq:H_P-barH}
  \bar{H} = H_{P} \comp \hl^{P},
  \qquad
  H_{P} = \bar{H} \comp \rho,
\end{equation}
and also that if $(X_{P}, dH_{P}) \in D_{P}$, then defining $\bar{X}(\bar{z}) \defeq T_{z}\rho \cdot X_{P}(z)$, we have
\begin{equation*}
  f(X_{P}(z), dH_{P}(z)) = \parentheses{ \bar{X}(\bar{z}), d\bar{H}(\bar{z}) },
\end{equation*}
because, using $\hl^{\mathcal{H}}_{z} \comp T_{\bar{z}}\varphi = (T_{z}\rho|_{\mathcal{H}_{z}})^{-1}$ (see diagram~\eqref{dia:Trho}) and Eq.~\eqref{eq:H_P-barH}, for any $v_{\bar{z}} \in T_{\bar{z}}T^{*}\bar{Q}$,
\begin{align*}
  \ip{ T^{*}_{\bar{z}}\varphi \comp (\hl^{\mathcal{H}}_{z})^{*} dH_{P}(z) }{ v_{\bar{z}} }
  &= \ip{ dH_{P}(z) }{ \hl^{\mathcal{H}}_{z} \comp T_{\bar{z}}\varphi(v_{\bar{z}}) }
  \\
  &= \ip{ dH_{P}(z) }{ (T_{z}\rho|_{\mathcal{H}_{z}})^{-1}(v_{\bar{z}}) }
  \\
  &= \ip{ \rho^{*}d\bar{H}(z) }{ (T_{z}\rho|_{\mathcal{H}_{z}})^{-1}(v_{\bar{z}}) }
  \\
  &= \ip{ d\bar{H}(\bar{z}) }{ T_{z}\rho \comp (T_{z}\rho|_{\mathcal{H}_{z}})^{-1}(v_{\bar{z}}) }
  \\
  &= \ip{ d\bar{H}(\bar{z}) }{ v_{\bar{z}} }.
\end{align*}
Therefore, the constrained implicit Hamiltonian system,
Eq.~\eqref{eq:ConstrainedIHS}, reduces to
\begin{equation*}
  (\bar{X}, d\bar{H}) \in \bar{D},
\end{equation*}
or
\begin{equation*}
  i_{\bar{X}} \bar{\Omega}^{\rm nh} = d\bar{H}.
\end{equation*}

\begin{remark}
  Again, this result is essentially a restatement of the nonholonomic reduction of \citet{Ko1992} (see also \citet{BaSn1993}, \citet{CaLeMaDi1999}, and \citet{HoGa2009}) in the language of Dirac structures and implicit Hamiltonian systems.
\end{remark}

\bibliography{DiracHJ}

\end{document}